\documentclass[prd,aps,amsfonts,eqsecnum,superscriptaddress,nofootinbib,longbibliography,notitlepage]{revtex4-1}

\usepackage{graphicx}
\usepackage{xcolor}
\usepackage{subcaption}
\usepackage{rotating}
\usepackage{amsmath,amssymb,graphics,amsthm,isomath}
\usepackage{physics }
\usepackage{verbatim}
\usepackage{circuitikz}
\usepackage{tikz}
\usepackage{esint}
\usepackage{blkarray}
\usepackage{algorithm}
\usepackage{algpseudocode}

\algdef{SE}[DOWHILE]{Do}{doWhile}{\algorithmicdo}[1]{\algorithmicwhile\ #1}
\algrenewcommand{\algorithmicreturn}{\State \textbf{return}}

\usepackage{xcolor}
\definecolor{mycitecolor}{rgb}{0.0, 0.45, 0.85}   

\usepackage[
  colorlinks=true,
  linkcolor=mycitecolor,
  citecolor=mycitecolor,
  urlcolor=mycitecolor,
  hyperindex=true,
  linktocpage=true
]{hyperref}

\usepackage[capitalise,compress]{cleveref}

\newcommand\andy[1]{{[\color{blue} andy: #1]}}
\newcommand\chao[1]{{[\color{magenta} chao: #1]}}
\newcommand\ben[1]{{\color{green!50!black}[Ben: #1]}}

\newcommand\carolyn[1]{{[\color{mylinkcolor} carolyn: #1]}}

\newcommand\Vol{\operatorname{vol}}
\newcommand\supp{\operatorname{supp}}
\newcommand\LRvel{v_{\rm LR}}
\newcommand\CLR{c_{\rm LR}}
\newcommand{\e}{\mathrm{e}}
\newcommand{\ii}{\mathrm{i}}

\renewcommand\equiv{:=}
\renewcommand\epsilon{\varepsilon}

\newtheorem{thm}{Theorem}
\numberwithin{thm}{section}
\newtheorem{cor}[thm]{Corollary}
\newtheorem{lem}[thm]{Lemma}

\newtheorem{prop}[thm]{Proposition}

\newtheorem{defn}[thm]{Definition}

\makeatletter
\renewcommand{\p@subsection}{}
\renewcommand{\p@subsubsection}{}
\makeatother

\usepackage{mathtools}
\tikzstyle{densely dashed}=          [dash pattern=on 4pt off 3pt]

\newcommand{\ad}{\operatorname{ad}}
\newcommand{\boxx}{b}

\usepackage{dsfont}

\begin{document}

\title{Lieb-Robinson bounds with exponential-in-volume tails}

\author{Ben T. McDonough}
\affiliation{Department of Physics and Center for Theory of Quantum Matter, University of Colorado, Boulder CO 80309, USA}

\author{Chao Yin}
\affiliation{Department of Physics and Center for Theory of Quantum Matter, University of Colorado, Boulder Colorado 80309, USA}

\author{Andrew Lucas}
\email{andrew.j.lucas@colorado.edu}
\affiliation{Department of Physics and Center for Theory of Quantum Matter, University of Colorado, Boulder Colorado 80309, USA}

\author{Carolyn Zhang}
\affiliation{Department of Physics, Harvard University, Cambridge, Massachusetts 02138, USA}

\begin{abstract}
 Lieb-Robinson bounds demonstrate the emergence of locality in many-body quantum systems.  Intuitively, 
 Lieb-Robinson bounds state that 
 with local or exponentially decaying interactions, the correlation that can be built up between two sites separated by distance $r$ after a time $t$ decays as $\exp(vt-r)$, where $v$ is the emergent Lieb-Robinson velocity.   In many problems, it is important to also capture how much of an operator grows to act on $r^d$ sites in $d$ spatial dimensions.   Perturbation theory and cluster expansion methods suggest that at short times, these volume-filling operators are suppressed as $\exp(-r^d)$ at short times.  We confirm this intuition, showing that for $r > vt$, the volume-filling operator is suppressed by $\exp(-(r-vt)^d/(vt)^{d-1})$. This closes a conceptual and practical gap between the cluster expansion and the Lieb-Robinson bound.  We then present two very different applications of this new bound.  Firstly, we obtain improved bounds on the classical computational resources necessary to simulate many-body dynamics with error tolerance $\epsilon$ for any finite time $t$: as $\epsilon$ becomes sufficiently small, only $\epsilon^{-\mathrm{O}(t^{d-1})}$ resources are needed.  A protocol that likely saturates this bound is given. Secondly, we prove that disorder operators have volume-law suppression near the ``solvable (Ising) point" in quantum phases with spontaneous symmetry breaking, which implies a new diagnostic for distinguishing many-body phases of quantum matter.
\end{abstract}

\maketitle

\tableofcontents

\section{Introduction}
The Lieb-Robinson theorem \cite{Lieb1972} proves that quantum correlations and entanglement spread with at most a finite velocity in many-body quantum lattice models.  While the original proof of this theorem is over 50 years old, it has recently become an extremely important technical tool in mathematical many-body physics \cite{AnthonyChen:2023bbe}.  The Lieb-Robinson theorem underlies proofs of (\emph{1}) the efficient simulatability of quantum dynamics on classical or quantum computers \cite{osborne2006,haah2021}; (\emph{2}) lower bounds on the time needed to prepare entangled states in quantum information processors \cite{Bravyi2006,entan_1d06}, including those with power-law interactions \cite{power_dyn_area17,power_chen19,power_simu19,power_KSLRB20,lucasprx2020,power_GHZ21,power_yifan21,power_KSOTOC21,power_LRB21,Wprotocol_gorshkov20,power_all2all24} and, in some cases, bosonic degrees of freedom \cite{boson_anharmonic08,boson_empty11,boson_spin13,boson_kuwahara21,boson_finitespeed22,boson_lemm22,boson_empty22,boson_algebraic24,boson_macro24,boson_longrange24}; (\emph{3}) the prethermal robustness of gapped phases of matter \cite{our_preth23} and the non-perturbative metastability of false vacua \cite{metastable24}; (\emph{4}) the stability of topological order \cite{hastings2005,topo_Hastings10,michalakis2013stability} and quantization of Hall conductance \cite{hastings2015}; (\emph{5}) the area-law of entanglement entropy in one dimension \cite{hastings2007} (although there are also combinatorial proofs); (\emph{6}) clustering of correlations in gapped ground states \cite{hastings2006,nachtergaele2006}; (\emph{7}) robustness of quantum metrology \cite{metro_HL24,metro_noisy24} (see Ref.~\cite{AnthonyChen:2023bbe} for a more complete list of applications). Besides these mathematical results, the intuition gained from the Lieb-Robinson bound has been important in developing a new theory of many-body quantum chaos in lattice models, where the onset of chaotic behavior at early times is characterized by the growth of a Heisenberg-evolved operator from a short string of Pauli matrices to a long string \cite{nahum2018,keyserlingk2018}.  This growth is controlled by a similar ``Frobenius light cone", which generally has a smaller velocity than the Lieb-Robinson light cone \cite{lucasprx2020}.   For this reason, insight from a precise understanding of quantum operator growth may lead to improved classical algorithms to simulate hydrodynamics \cite{pollmannhydro}.  

This diversity of applications, spanning quantum information sciences, atomic physics, condensed matter, and even high-energy physics, usually relies on a Lieb-Robinson theorem stated as follows:  given two local operators $A_x$ and $B_y$, and a local many-body lattice model, \begin{equation}
    \lVert [A_x(t),B_y]\rVert \lesssim \exp [\mu(vt - \mathsf d(x,y))], \label{eq:introLR}
\end{equation}
where $\mathsf d(x,y)$ denotes the distance between the degrees of freedom and $\mu$ is a constant.  This bound is sufficient for many of the applications listed above. To give one example of how the simple bound \eqref{eq:introLR} would be used, let us briefly discuss how to bound the classical simulatability of quantum dynamics. Suppose that we wish to study the Heisenberg-evolved operator $A_x(t)$ by solving the Heisenberg equation of motion for $A_x$.  Given finite computational resources, we simply truncate the list of Pauli strings we keep track of to those supported in a ball of radius $R$ around site $x$. We will refer to this truncated operator as the ``fraction" of the operator acting within this ball. In $d$ spatial dimensions, the number of such operators scales as $N\sim \exp[R^d]$.  Equation \eqref{eq:introLR} suggests that the error in this approximation $\epsilon$ will be bounded by $\epsilon \lesssim \exp[vt-R]$.
The classical resources necessary to simulate dynamics can thus be estimated as  
\begin{equation}
    N\sim \exp\left[\left(vt + \log \epsilon^{-1}\right)^d \right]. \label{eq:introsim1}
\end{equation}
For the rest of the introduction, O(1) prefactors in scaling relations will be suppressed.

The argument above---as do many other applications of a Lieb-Robinson bound---crucially relies on the \emph{exponential decay in distance} in Eq. \eqref{eq:introLR}. Is that optimal?  As phrased in Eq. \eqref{eq:introLR}, the bound is pretty much optimal: the tail bound can be improved to at most $\exp(-R\log R)$ with strictly local interactions \cite{AnthonyChen:2023bbe}.   This can be intuitively seen by noting that at order $R$, the operator $A_x(t)$ can grow $R$ sites away: \begin{equation}
    A_x(t) = A_x + \mathrm{i}t[H_{x+1,x}, A_x] + \cdots + \frac{(\mathrm{i}t)^R}{R!}[H_{x+R,x+R-1}, [\cdots, [H_{x+1,x},A_x]]] + \cdots.   \label{eq:taylorseriesintro}
\end{equation}
On the other hand, when we bounded simulatability in $d>1$, we used a Lieb-Robinson bound for the spreading of an operator from $x$ to $y$, and assumed that this same error controls how much of the operator might act on the whole ball of radius $\mathsf d(x,y)$.  This approximation might seem loose.  Indeed, Eq. \eqref{eq:taylorseriesintro} suggests that the first terms that act on an entire ball of radius $R$ arise at order $R^d$, so we might expect a tail bound that is suppressed in \emph{volume}:  $\exp(-(R-vt)^d)$.  Lieb-Robinson bounds of this kind are not known.

This understanding is physically important because it sharpens our picture of operator growth outside of the LR light cone. This has important applications to the simulability of continuous-time dynamics and to extended operators such as disorder operators, which are used broadly in the literature. To illustrate the simulability application, if $A$ is a local operator and we study the Heisenberg time-evolved $A(t)$ expanded in the basis of Pauli strings, then a volume-law scaling bound suggests that the minuscule weight of volume-filling Pauli strings could overcome the large number of them---$\exp[\mathrm{O}(t^d)]$---to account for.  We might then hope for an algorithm with an improved \emph{polynomial} dependence in $\varepsilon^{-1}$:
\begin{equation}
N \lesssim \exp[(vt)^{d-1}(vt + \log \epsilon^{-1})].
\label{eq:introsim3}
\end{equation}
Here $(vt + \log \epsilon^{-1})$ is the radius corresponding to an error tolerance $\epsilon$ and the prefactor $(vt)^{d-1}$ is the surface area of the light cone, which scales with the number of paths [as in Eq. \eqref{eq:taylorseriesintro}] that reach the edge of the light cone.

In fact, using a very different approach---cluster expansions---it has recently been shown \cite{cluster_Alhambra23} that classical simulation algorithms do exist with polynomial error in $\epsilon^{-1}$ at a fixed time $t$:
\begin{equation}
    N \sim \exp\left(\mathrm{e}^t \log\epsilon^{-1}\right). \label{eq:introsim2}
\end{equation} 
Despite scaling optimally with $\epsilon$, this bound is weaker than Eq. \eqref{eq:introsim1} for large $t$ and fixed $\epsilon$. The mismatch in scaling with $t$ and $\epsilon$ between Eqs. \eqref{eq:introsim1} and \eqref{eq:introsim2} would be closed by a bound such as in Eq. \eqref{eq:introsim3}. To find such an optimal bound, which may also find many other applications, it is crucial to have better control over the shape of $A_x(t)$ outside of the ``Lieb-Robinson light cone"---the region within distance $vt$ of the initial site $x$. This paper will address precisely this problem, and discuss the extent to which Eq. \eqref{eq:introsim3} might be achievable.

\section{Summary of results}
Our central objective is to better understand the ``tail" of an operator---the fraction that acts outside the light cone. Quantum circuits with local gates have strict light cones, so time-evolved 
operators have no tails. This is one of the few cases where we cannot rely on intuition from quantum circuits \cite{nahum2018,keyserlingk2018} to inform 
us about the behavior of operators in continuous-time evolution. Intuitively, from Eq. \eqref{eq:taylorseriesintro}, it appears that the exponential tail in the Lieb-Robinson bound essentially arises due to \emph{direct paths} between two points $x$ and $y$. With a local Hamiltonian, terms arising from direct paths between points separated by a distance $R$ grow in a sequential fashion, requiring $\mathrm \Omega(R)$ steps.  In contrast, if we want to fill the entire volume of sites a distance $R$ from $x$, we should need $\mathrm \Omega(R^d)$ terms in the series expansion \eqref{eq:taylorseriesintro} to ensure that all sites are touched at least once.   This argument suggests that the weight of an operator that touches a finite fraction of the sites inside a ball of radius $R$ should decay as $\exp[-R^d]$ rather than $\exp[-R]$, and that the Pauli strings grown along direct paths which dominate the behavior outside the light cone have a thin, ``noodle-shaped" support  (illustrated in Figure \ref{fig:schematic}a).
\begin{figure}[t]
\centering
\includegraphics[width = .6\textwidth]{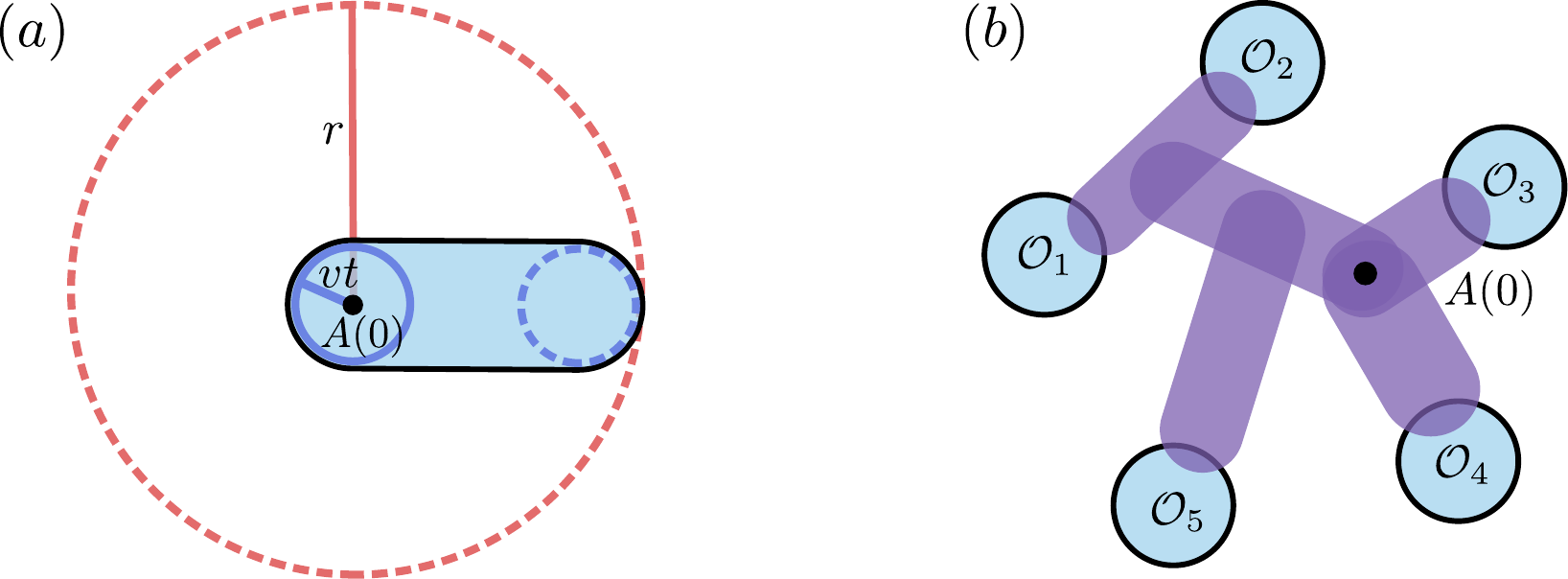}
\caption{(\emph{a}) The dominant contribution to the growth of an operator outside of the Lieb-Robinson light cone.  Conventional Lieb-Robinson bounds suggest that the fraction of $A(t)$ acting on the red dotted circle should be suppressed by $\exp[-\mathrm{O}(r)]$. We show that it is suppressed by $\operatorname{exp}[-\mathrm{O}(r^d)]$, which highly favors operators with ``noodle-shaped" support like that shown in blue. (\emph{b}) An illustration of why this problem is so difficult to approach with conventional techniques, which count the number of paths of subsets (illustrated in purple) that connect the support of $A(0)$ to the support of $\mathcal O_1, \dots ,\mathcal O_m$. The combinatorics of counting branching paths that grow operators to have large support becomes difficult at higher orders. This makes the equivalence-class method a very valuable tool. }
\label{fig:schematic}
\end{figure}
In $d>1$, we must be careful about predicting the scaling directly from the number of terms in the perturbative expansion \eqref{eq:taylorseriesintro}.  As depicted in Figure \ref{fig:schematic}b, the number of ways that an operator can grow to fill a large volume is exponentially large in volume itself!  As we Taylor expand $A_x(t)$ at higher orders $n$ in $t$,
the larger support of $[H_{X_{n-1}},\cdots,[H_{X_1},A_x]]$ means that there are yet more terms in $H$ that may not commute with the nested commutator at the $n^\text{th}$ step.  As is known \cite{araki,Avdoshkin:2019trj}, the growth in the number of possible terms is enough to break the convergence of the series.  Therefore, if the weight of an operator that fills an entire volume truly does decay as $\exp[-R^d]$, we must find a careful way to re-sum this series.  

Our main technical tool for this re-summation is the equivalence-class formalism recently introduced in Ref.~\cite{chen2021operator}. In Sec.~\ref{subsec:advantage}, we use an example of a system with a weak link to illustrate the flexibility in choosing which paths to include in the bound. In Sec.~\ref{sec:nestedcommutatorbounds} 
we leverage this flexibility to introduce a reformulation of the problem, as illustrated in Fig.~\ref{fig:setup}.
This reformulation allows us to separate out the direct paths connecting sites within a ball of radius $R$ to $x$. We phrase our theorem in terms of nested commutators as a tool for probing the support of an operator.
\begin{figure}[t]
    \centering
    \includegraphics[width=0.3\linewidth]{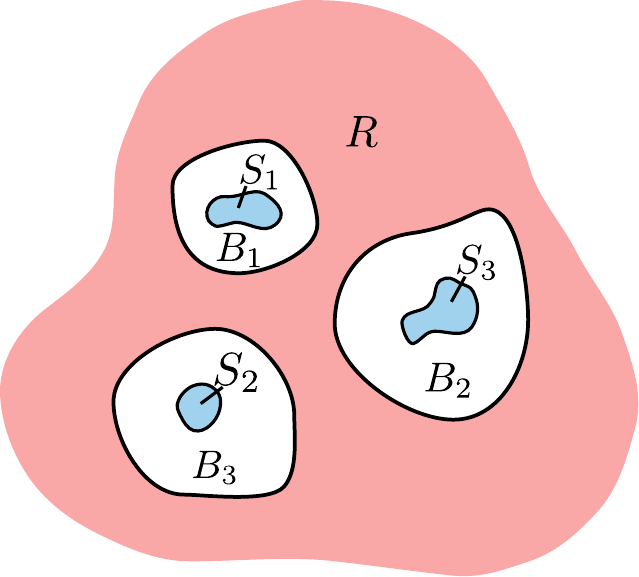}
    \caption{Illustration of the problem setup. This illustration is schematic, and the underlying vertex set is not shown. Region $R$ is depicted in red, and can be imagined as having infinite extent. Regions $S_1, S_2, S_3$ are depicted in blue, contained within the white regions $B_1, B_2, B_3$ which are disjoint from $R$. The shapes are arbitrary, but, for particular applications, we find optimal choices for the shape of these regions.}
    \label{fig:setup}
\end{figure}

We use ``local" to refer to Hamiltonians with finite-range interactions and ``quasi-local" to refer to Hamiltonians with interactions that decay exponentially and sufficiently rapidly with the volume of their support, which may be slightly stronger than the way it has been used before in the literature.
Our first result is to rigorously prove the following:
\begin{thm}[Thm.~\ref{thm:bound} and Thm.~\ref{thm:exponential_H} (Informal)]\label{thm:introvolumelaw}
Suppose that $S_1, \dots, S_m$ are regions which are at least $r$ apart. Let $\mathcal O_1, \dots ,\mathcal O_m$ be operators that only act non-trivially on these regions respectively, with operator norm $\Vert \mathcal O_i \Vert = 1$ for simplicity. Let $A$ be an operator with $\Vert A \Vert = 1$, and suppose that the distance between the support of $A$ and any $S_i$ is also at least $r$. If $H$ is a local Hamiltonian, then there exist constants $c_{\text{LR}}, v_{\text{LR}}$ independent of $m$ such that
\begin{align}
\Vert [\mathcal O_m, [\dots,[\mathcal O_1, A(t)]]] \Vert \leq \CLR^m
\qty(\frac{\LRvel t}{r})^{mr}
\end{align}
for any $\LRvel t < r$.
If $H$ is quasi-local, then as long as $S_i$ are not ``too close" together, the following slightly weaker bound holds:
\begin{align}
\Vert [\mathcal O_m, [\dots,[\mathcal O_1, A(t)]]] \Vert \leq \CLR ^m\exp(m\mu(\LRvel t-r))
\end{align}
with $\mu$ a constant independent of $m$. The value of $\CLR$ is irrelevant as long as $r$ is sufficiently large.
\end{thm}

Intuitively, the nested commutator is a probe of the support of $A(t)$ because the operators $\mathcal O_1, \dots, \mathcal O_m$ supported on sets $S_1, \dots, S_m$ may be chosen to capture the contribution of Pauli strings in the expansion of $A(t)$ which act simultaneously on $S_1, \dots, S_m$. At a high level, it certainly would seem that Thm.~\ref{thm:introvolumelaw} is a dramatic improvement over a conventional Lieb-Robinson bound.  Rather than scaling with the diameter of the support, the suppression depends on the volume. Just like conventional Lieb-Robinson bounds, these results also hold for time-dependent systems; the intuition behind this generalization is that these worst-case bounds are established by counting intersections, and this combinatorial argument still works when time-dependence is introduced. We show how the proof can be generalized in Sec.~\ref{sec:time-dependent}.

In some applications, the nested commutator itself is already an interesting object to bound. However, for applications such as bounding the simulatability of quantum dynamics on classical computers, directly characterizing the ``fraction" of an operator that acts on a subset lying far outside the light cone is crucial. Specifically, an operator under time evolution becomes a complicated weighted sum of operators of different support, and such a characterization helps us design an efficient classical algorithm that keeps track of only the most important contributions to the sum.
 
To make the connection between the nested commutator bound and volume-law scaling precise, we observe that $m\sim \left(R/vt\right)^d$ such operators $\{\mathcal O_i\}_{i=1}^m$ can fit inside of a ball of radius $R$, while ensuring that no light cones overlap. This leads to the following theorem.
\begin{thm}[Cor.~\ref{cor:volumescaling}, Cor.~\ref{cor:volumescaling_exp}, and Lemma~\ref{lem:ASb<} (Informal)]
Given a local or quasi-local interaction and a local operator $A$ on a cubic lattice, consider the Pauli decomposition $A(t) = \sum_{S\subseteq V}\sum_{i=1}^{3^{|S|}}\alpha_S^{(i)}(t)P_{S}^{(i)}$, where $P_{S}^{(i)}$ is a tensor product of Pauli operators supported on subset $S$ of vertices and $\{\alpha_S^{(i)}(t)\}$ are complex coefficients. 
Tile the square lattice with Euclidean boxes of side length sufficiently greater than $\LRvel t$ such that the support of $A$ is contained in a single box. Given a connected set of boxes $(B_1, \dots, B_m)$ called a ``cluster," where each box has side length $r$ sufficiently greater than $\LRvel t$, there exists $\mu, \CLR$ such that
\begin{align}
\left \Vert\sum_{B_k \cap S \neq \emptyset,i}\alpha_S^{(i)}(t)P_{S}^{(i)}\right\Vert \leq \CLR\exp(-\mu r[3^{-d}m-1]),
\end{align}
where the sum over $S$ is taken over subsets which intersect $B_1, \dots, B_m$ non-trivially.
\end{thm}
We expect that the restriction to a cubic lattice is just a technical convenience, and that a similar result should hold for a general graph embedded in $d$ spatial dimensions.

The first application of this volume-law scaling bound is given in Sec.~\ref{sec:simulation}, where we explicitly derive an improved estimate for the simulation complexity of local and quasi-local quantum systems.
\begin{thm}[Cor.~\ref{cor:polynomial_simulate} (Informal version)]
Let $H$ be a local or quasi-local Hamiltonian on a $d-$dimensional cubic lattice of qudits. If $\rho$ is a state such that the marginal $\rho_S$ on any connected set $S$ can be obtained with at most exponential resources, then we provide an algorithm to compute the expectation $\Tr[A(t)\rho]$ to error $\epsilon$ with classical resources
\begin{equation}
N = \exp(\mathrm O[(1+\LRvel t)^{d-1}\qty(\LRvel t + \log \epsilon^{-1})]).
\end{equation}
\end{thm}
This complexity exhibits polynomial scaling in $1/\epsilon$, and we argue that this bound has the best possible scaling in both $\epsilon$ and $t$. Note that this scaling is just slightly weaker than the naive expectation of Eq.~\eqref{eq:introsim2}.  In particular, this result recovers the $\exp(\mathrm O((\LRvel t)^d))$ time-dependence that one would expect from a simulation with a strict light cone, such as a random unitary circuit \cite{nahum2018,keyserlingk2018}, while maintaining the polynomial scaling in error tolerance $\epsilon$ for a fixed time with a suitably $d$- and $t$-dependent prefactor, which reduces to a constant in the case $d=1$, as known from existing Lieb-Robinson bounds.

We additionally explore the applications of bounds such as Eq.~\eqref{thm:introvolumelaw} in condensed matter physics. More concretely, we study  spontaneous $\mathbb{Z}_2$-symmetry-breaking phases of quantum matter, and rigorously close two gaps between known characterizations of symmetry-breaking phases and expectations from numerics and perturbation theory.
First, conventional LR bounds \cite{hastings2005} imply that the splitting $\delta = |\bra{\psi_+}H\ket{\psi_+} - \bra{\psi_-}H \ket{\psi_-}|$ between the two symmetry-broken ground states $\ket{\psi_{\pm}}$ in the $\pm 1$ sectors of the global $\mathbb{Z}_2$ symmetry is at most exponentially small in the linear system size. However, the expectation from perturbation theory is that this splitting is exponentially small in system \textit{volume}. Spontaneous symmetry breaking is also characterized by the rapid decay of the disorder operator $D_R \equiv \prod_{i \in B_R(v)}X_i$ with $R$, which is the Ising symmetry restricted to a ball of radius $R$. Conventional LR bounds give a decay which is exponential in $R$, while the expectation from numerics and perturbation theory \cite{zhao2021} is that this decay is exponential in the volume $\sim R^d$. We close these gaps with the following theorem.
\begin{thm}[Thm.~\ref{thm:splitting} and Thm.~\ref{thm:disorder} (Informal)]
If $\ket{\psi_{\pm}}$ are the ground states of a gapped, quasi-local Hamiltonian in $d$ dimensions which are connected to GHZ states with eigenvalues $\pm 1$ under the global $\mathbb{Z}_2$ symmetry by a finite-time evolution under a quasi-local Hamiltonian, then the splitting obeys the following bound
\begin{equation}
|\bra{\psi_+}H\ket{\psi_+} - \bra{\psi_-}H\ket{\psi_-}| \leq \CLR L^d\exp(-\mu \frac{(L-\LRvel t)^d}{(\LRvel t)^{d-1}}),
\end{equation}
where $L$ is the linear size of the system and $\CLR, \mu, \LRvel$ are constants.
Furthermore, if $\ket{\psi}$ is connected to the polarized state $\ket{\vb 0}$ by a finite-time evolution under a $\mathbb Z_2$-symmetric, quasi-local Hamiltonian, then the disorder parameter obeys
\begin{align}
\bra{\psi}D_R\ket{\psi} \leq \CLR\exp(-\mu \frac{(R-\LRvel t)^d}{( \LRvel t)^{d-1}}).
\end{align}
\end{thm}
These results suggest a new diagnostic tool to distinguish quantum phases of matter, which we illustrate via a concrete example with Rokhsar-Kivelson-like states \cite{rokhsar1988,ardonne2004,castelnovo2008}.

\section{Mathematical preliminaries}
The remainder of the paper will discuss the above results in a more formal way.  We first introduce our notation and review some key results and ideas from previous work on Lieb-Robinson bounds.  We collect such results in this section.
\subsection{Many-body quantum systems}
To discuss locality in a many-body quantum system, it is often helpful to associate the quantum degrees of freedom (``qudits") with the vertices of a graph.  We can imbue the problem with a notion of spatial locality by adding edges between these vertices, providing a notion of distance between distinct qubits.  We begin the formal discussion with problems where qudits only interact with their nearest neighbors on the resulting graph (with generalizations to exponentially decaying interactions discussed in Sec.~\ref{sec:exponential}), and in this context it is often helpful to define a factor graph \cite{chen2021operator}.
\begin{defn}[Factor Graph]
A factor graph $G = (V, E, F)$ is a bipartite graph in which each node in the \textit{node set} $V$ is connected to the \textit{factor set} $F$ with \textit{edges} $E\subseteq V\times F$.   We assume that $G$ is  connected.
If $(v,X) \in E$, we write $v \in X$: vertex $v\in V$ is connected to factor $X\in F$.
The distance $\mathsf{d}(x,y)$ between $x,y\in V$ is defined as the smallest number of factors contained in a connected path from $x$ to $y$ in $G$.   For subsets $R,S\subset V$, $\mathsf{d}(R,S) := \min_{x\in R, y\in S} \mathsf{d}(x,y)$.
\end{defn}
\begin{defn}
    Given a subset $S \subset V$, define the boundary set \begin{equation}
        \partial S := \lbrace x \in S : \mathsf{d}(\lbrace x\rbrace,S^{\mathrm{c}}) = 1\rbrace.
    \end{equation}
\end{defn}
The tail bounds we develop fully leverage the dimensionality of the system. To capture the dimensionality of an arbitrary graph, we use the following notion:
\begin{defn}[$d$-dimensional system]
     Given a factor graph $(V,E,F)$, define the ball of radius $r$ around vertex $v\in V$ as \begin{equation}
        B_r(v) := \lbrace u\in V : \mathsf{d}(u,v)\le r\rbrace.
    \end{equation} We say that $(V,E,F)$ is d dimensional if for all $v\in V$ and sufficiently large $r$, there exist O(1) constants $C_{1,2}$ and $C_{1,2}'$ such that \begin{subequations}\begin{align}
        C_1'r^d \le  |B_r(v)| &\le C_1 r^d, \\
        C_2'r^{d-1} \le |\partial B_r(v)| &\le C_2 r^{d-1}.
    \end{align}\end{subequations}
\label{def:distance}
\end{defn}
Finite factor graphs represent a useful language to speak about many-body systems, as formalized in the definition below.   
We demonstrate all our results where $V$, $E$, and $F$ are finite sets, but our formalism naturally extends to the thermodynamic limit.

\begin{defn}[Many-body Quantum System with Nearest-Neighbor Interactions]
    Given a factor graph $G = (V, E, F)$, we define a many-body Hilbert space \begin{equation}
        \mathcal{H} = \bigotimes_{v\in V} \left(\mathbb{C}^q\right)_v =: \left(\mathbb{C}^q\right)^{\otimes V},
    \end{equation}
    for some $q\in\mathbb{N}$. Here $\left(\mathbb{C}^q\right)_v$ reminds us that the $q$-dimensional Hilbert space is associated with a degree of freedom at vertex $v$. The Hamiltonian is a Hermitian operator \begin{equation}
        H = \sum_{X\in F} H_X,
    \end{equation}
    where $X$, considered as a subset of $V$, is the support of $H_X$.  We say that $\mathcal H, H$ defines a many-body system with few-body interactions if for all $X\in F$, $|X| \leq \Delta$ for some fixed $\mathrm O(1)$ constant $\Delta$.
\end{defn}

\begin{defn}[Vector space of operators]
The space of all linear operators acting on $\mathcal{H}$ forms a vector space $\mathcal{B}$. When we wish to emphasize this, we will write $A\in \mathcal{B}$ as $|A)$. 
There is a natural adjoint action of any $A\in \mathcal{B}$ on $\mathcal{B}$:  $\mathrm{ad}_A |B) = |[A,B])$. 
We often use $A_S\in \mathcal B_S$ to denote that $A$ is supported within $S\subset V$.
\end{defn}
We can express the time evolution of an operator $\mathcal O$ in the Heisenberg picture using the notation above:
\begin{equation}
\dv{t}|\mathcal O(t)) = \mathrm{i}|[H, \mathcal O(t)]) := \mathcal L |\mathcal O(t)) \ .
\label{eq:HL}
\end{equation}
The formal solution to this equation for a time-independent $H$, and thus a time-independent $\mathcal L$, is $|\mathcal O(t)) = \mathrm{e}^{t\mathcal L}|\mathcal O)$. Since $H = \sum_{X \in F}H_X$, we can similarly associate a superoperator $\mathcal L_X := \mathrm{i}\ad_{H_X}$ with each $H_X$ and, through the bilinearity of the commutator, write $\mathcal L = \sum_{X \in F}\mathcal L_X$. For most of this paper, for notational convenience, we restrict to the study of $t$-independent $H$. In Sec.~\ref{sec:time-dependent}, we show how propagators are replaced with the appropriate time-ordered exponentials, and the rest of the manipulations are largely unchanged.

\subsection{Equivalence-class construction of Lieb-Robinson bounds}
A typical Lieb-Robinson bound is stated as follows:
\begin{thm}[Lieb-Robinson bound]
\label{thm:standardLR} 
Given a Hamiltonian $H$ on factor graph $G=(V,F,E)$, there exist $\mathrm O(1)$ constants $C, \mu, \LRvel$ such that for any two single-site operators $A_{R,S}\in \mathcal{B}_{R,S}$ respectively, for subsets $R,S\subset V$, \begin{equation}\label{lrbound1}
    \lVert [A_R(t),A_S] \rVert \le C \min(|\partial R|, |\partial S|) \mathrm{e}^{\mu (vt-\mathsf{d}(R,S))}.
\end{equation}
Here $\lVert \cdot \rVert$ denotes the operator norm (maximum singular value of its argument).
\end{thm}
Despite its seemingly technical nature, this result has broad applications (see Ref.~\cite{AnthonyChen:2023bbe} for a recent review).  Although the Lieb-Robinson bound (\ref{lrbound1}) is over 50 years old \cite{Lieb1972},\footnote{We note, however, that the optimal $\min(|\partial R|, |\partial S|)$ prefactor in \eqref{lrbound1} has not always been employed in the literature.} we now focus on a much more recent \cite{chen2021operator} ``equivalence-class formulation" of Lieb-Robinson bounds, which we find will elegantly address the central question in this paper.  Here, we review this equivalence-class formulation with a focus on the high-level ideas; a formal proof of the results outlined below follows from the more general method we develop in the next section.

Given operators $A_i$ and $\mathcal O_j$ supported on sites $i$ and $j$ respectively, conventional Lieb-Robinson bounds deal with the commutator
\begin{equation}
C_{ij}(t) = \frac{1}{2}\left\Vert [\mathcal O_j, A_i(t)]\right\Vert  = \frac{1}{2}\left\Vert\ad_{O_j}|A_i(t))\right\Vert 
\ 
\end{equation}
analogously to Theorem \ref{thm:standardLR}, using single sites rather than sets for simplicity.   
We first express $|A_i(t))$ as:
\begin{equation}
\ad_{\mathcal O_j} e^{t\mathcal L}|A_i) = \sum_{n=0}^{\infty}\frac{t^n}{n!}\sum_{X_1, \dots, X_n \in F}\ad_{\mathcal O_{j}}\mathcal L_{X_n}\dots \mathcal L_{X_1}|A_i) \ .
\label{eq:resum}
\end{equation}
Each term in the sum can be uniquely labeled by a sequence of factors $M = (X_1, \dots,X_n)$.  
These sequences correspond to graphs whose nodes are elements of $F$, where $X_1, X_2$ are connected by an edge if they share at least one vertex, and the corresponding term is only non-vanishing if the graphs are connected. However, these graphs include redundancies, such as offshoots and cycles, which intuitively do not contribute to the growth of the operator from site $i$ onto site $j$. 
A fundamental innovation of Ref.~\cite{chen2021operator} is to formalize this intuition by bounding the commutator using only so-called \textit{irreducible} (non-self-intersecting) paths between $i$ and $j$ on the factor graph. These paths correspond to sequences of factors that appear in the expansion of the time evolution operator.
Two possible connected acyclic subgraphs of the factor graph which differ only by redundant ``reducible" terms but nevertheless have the same irreducible path are said to be in the same equivalence class. An example is illustrated pictorially in Fig.~\ref{fig:equivalence_class_example}.

\begin{figure}
\centering
\includegraphics[width = .65 \textwidth]{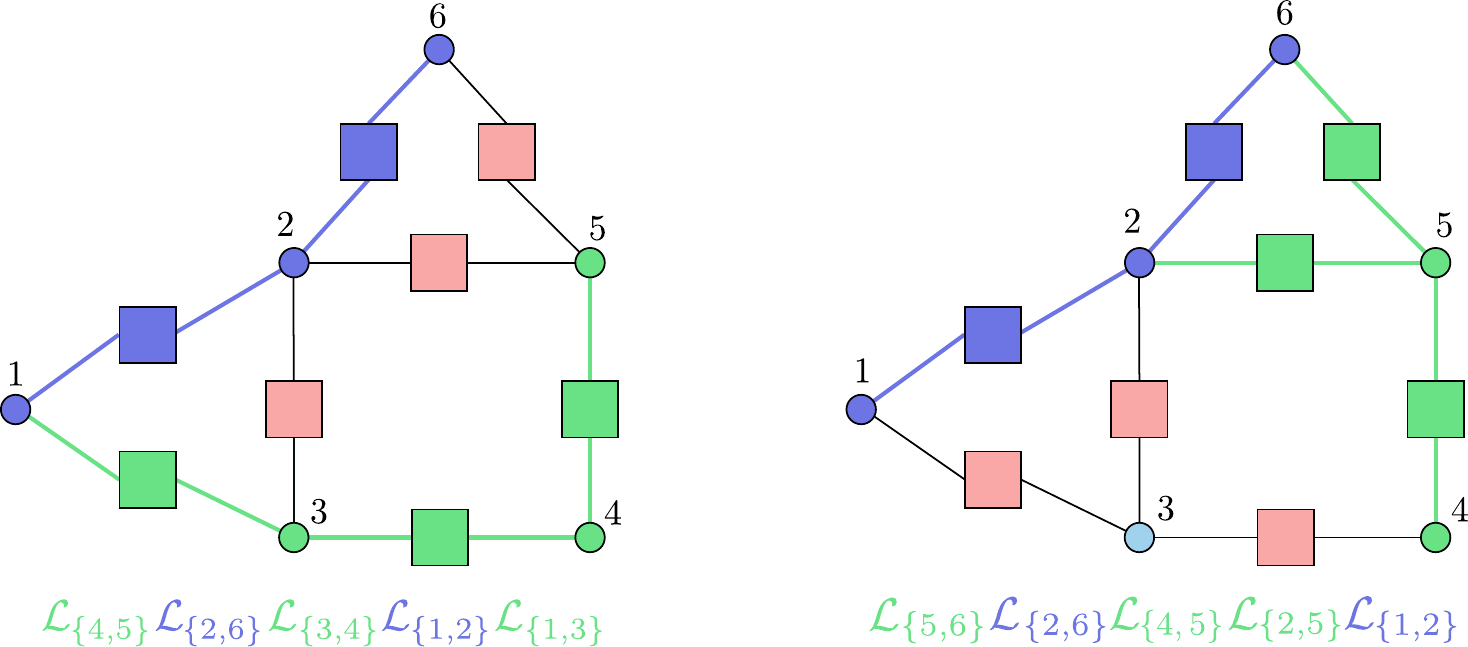}
\caption{Example of two sequences of factors that appear in the expansion of the Liouvillian propagator which have the same irreducible path and thus belong to the same equivalence class. The boxes are factors and the circles are vertices. The irreducible path is shown in blue, while the irrelevant terms in the sequence are shown in green. Note that the ordering of the terms in the sequence matters for the construction of the irreducible path; for instance, if $\mathcal L_{\{2,6\}}$ is to appear after $\mathcal L_{\{5,6\}}$ in the right figure, then the irreducible path would be $\{1,2\}, \{2,5\}, \{5,6\}$.}
\label{fig:equivalence_class_example}
\end{figure}
Grouping the terms in Eq.~\eqref{eq:resum} by equivalence class, Chen and Lucas obtained the following improved LR bound.
\begin{thm}[Theorem 3 of Ref.~\cite{chen2021operator}]
\label{thm:chen_LR_bound}
If $\mathcal O_{j} \in \mathcal{B}_j$ and $A_i\in\mathcal{B}_i$ with $\Vert \mathcal O_j\Vert =  \Vert A_i\Vert = 1$, then
\begin{equation}
C_{ij}(t) := \frac{1}{2}\Vert \ad_{\mathcal O_j} e^{\mathcal L t}|A_i)\Vert \leq \sum_{[\Gamma] \in \mathcal S}\frac{(2|t|)^{|\Gamma|}}{|\Gamma|!}w(\Gamma),
\end{equation}
where $w(\Gamma) = \prod_{X \in \Gamma}\Vert H_X \Vert$ is the weight of the path $\Gamma$, $[\Gamma]$ is an equivalence class of paths differing only by reducible terms, and $\mathcal S$ is the set of these equivalence classes.
\label{thm:eqconstruction}
\end{thm}

The first main contribution of this work is to use similar techniques to prove a generalization of the previous theorem to nested commutators. As such, we will not prove Theorem~\ref{thm:chen_LR_bound} here, and instead defer the details of the proof to Sec.~\ref{sec:nestedcommutatorbounds}. We also generalize the following corollary to Thm.~\ref{thm:chen_LR_bound}.
\begin{cor}[Cor. 6 of \cite{chen2021operator}]
\label{cor:exp_path_counting}
Define a symmetric real matrix $h$ by 
\begin{equation}
h_{ij} =
\begin{cases}
\sum_{X \in F: \{i,j\}\subseteq X}\Vert H_X \Vert & i \neq j \\
0 & i = j
\end{cases}
\end{equation}
Then 
\begin{align}
C_{ij}(t) \leq \exp(2h|t|)_{ij}
\end{align}
\end{cor}
Essentially, matrix $h$ is the weighted adjacency matrix of a graph that encodes the connectivity of the Hamiltonian. It is a well-known graph-theoretic result that the matrix element $\exp(h)_{ij}$ sums the weighted paths from $i$ to $j$. We likewise defer further details of the proof to Sec.~\ref{sec:nestedcommutatorbounds}.

\subsection{Advantage over conventional Lieb-Robinson bounds\label{subsec:advantage}}

There is a lot of freedom in choosing the equivalence classes in Thm.~\ref{thm:chen_LR_bound}, which gives us the ability to obtain tighter bounds by sorting equivalence classes with a specific problem in mind. Many of the bounds to be obtained in future sections revolve around different methods of constructing equivalence classes. In this section, we provide a motivating example of such a situation.

Consider a chain consisting of four sites (see Fig.~\ref{fig:example_chain}). The Hamiltonian is a sum of three terms $\{h_{1,2}, h_{2,3}, h_{3,4}\}$, with $\Vert h_{1,2} \Vert = \Vert h_{3,4} \Vert = h$ and a weak link $\Vert h_{2,3}\Vert = \epsilon h$ with $\epsilon \ll 1$. 
Let $A$ be an operator supported on site $1$ and $B$ supported on site $4$, both with unit norm. 
Using Theorem~\ref{thm:chen_LR_bound},
\begin{align}
\Vert [A, B] \Vert \leq 2\sum_{[\Gamma] \in S_{ij}}\frac{(2|t|)^{|\Gamma|}}{|\Gamma|!}\prod_{X \in \Gamma}\Vert H_X \Vert = \frac{2\epsilon(2h|t|)^3}{3!} \label{eq:epsilon13bound}
\end{align}
which follows because $\Gamma = (\{1,2\}, \{2,3\}, \{3,4\})$ is the only irreducible path connecting the two points.

 The above does not take advantage of the fact that we know where the weak coupling is in the system. We could instead consider equivalence classes labeled by irreducible paths connecting $\{1,2\}$ with $\{3,4\}$, which in this simple example is trivially the factor $\mathcal L_{2,3}$.   Similar ideas were described in \cite{commute_graph20,baldwin}. Applying the same theorem,
\begin{equation}
\Vert [A, B]\Vert \leq 4\epsilon h|t|
\end{equation}
If $\epsilon \ll 1$ then this bound is nontrivial (smaller than 1) out to later times $t\lesssim \epsilon^{-1}$, versus $t\lesssim \epsilon^{-1/3}$ for Eq.~\eqref{eq:epsilon13bound}. The physical intuition for why this bound works is that we are effectively working in the interaction picture with $H_0 = h_{12} + h_{34}$, rotating out the contribution of these terms to specifically target the weak link. 

\begin{figure}[t!]
    \centering
    \includegraphics[width=0.6\linewidth]{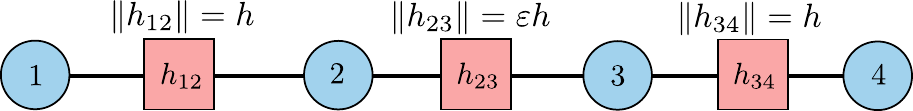}
    \caption{Example system showing the advantage to shortening the irreducible paths.}
    \label{fig:example_chain}
\end{figure}

\section{Nested commutator bounds\label{sec:nestedcommutatorbounds}}
Having reviewed the standard single-commutator Lieb-Robinson bound, we now develop a formalism to generalize Thm.~\ref{thm:chen_LR_bound} to nested commutators. The proof technique is quite similar; however, there are additional details to check, including a convenient choice of equivalence classes and the need to account for the ordering of terms between the different irreducible paths, which was automatic in the case of a single commutator.

This section focuses on systems with only nearest-neighbor interactions on some interaction graph.  The generalization to systems with exponential-tailed interactions (which is important for some of our applications) will be presented in Sec.~\ref{sec:exponential}.

\subsection{Equivalence classes of causal trees}\label{subsec:trees}
Let $H$ be a Hamiltonian associated with the factor graph $(F,E,V)$. Let $\mathcal O_1, \dots, \mathcal O_m$ be operators with $\Vert \mathcal O_i \Vert = 1$ and disjoint supports $S_1, \dots, S_m$, respectively, which are contained in regions $B_1, \dots, B_m$. 
Consider an operator $A$ with $\Vert A \Vert = 1$ whose support is a set $R$ contained within the complement of $B_1, \dots, B_m$, as shown in Fig.~\ref{fig:setup}. We are interested in bounding the nested commutator
\begin{equation}
\label{eq:nest_comm_def}
C_{\vec S}^R(t) = \sup_{\substack{\mathcal O_1, \dots \mathcal O_m\\ A}} \frac{1}{2^m}\Vert [\mathcal O_m, [ \dots ,[\mathcal O_1, A(t)]]\dots]\Vert
\end{equation}
Note that $\prod_{i=1}^m\ad_{\mathcal O_m}$ is a superoperator that acts trivially on any operator $A$ such that $\supp(A)$ (denoting the support) is disjoint from $\supp(\mathcal O_i)$ for each $i$. This property is independent of the ordering of the nested commutator due to the Jacobi identity and the observation that $[\ad_{\mathcal O_i}, \ad_{\mathcal O_j}] = \ad_{[\mathcal O_j, \mathcal O_j]} = 0$ by our assumption that the supports of $O_i, O_j$ are disjoint.

As illustrated before, we can rewrite the time-evolved operator $A$ as
\begin{align}
|A(t)) = \mathrm{e}^{\mathcal L t}|A) = \sum_{M}\frac{t^{|M|}}{|M|!}\mathcal L_{M_{|M|}}\dots \mathcal L_{M_1}|A)
\end{align}
where again, $M = (M_1, M_2, \dots, M_{|M|})$ is a sequence of factors. The essence of our approach is to choose the appropriate equivalence classes to quantify the way that $A(t)$ ``leaks" into the regions $B_1, \dots, B_m$ over time. 
To accomplish this, we construct an algorithm to map each $M$ to a graph on subsets of $V$, modifying the construction in Ref.~\cite{chen2021operator}. The algorithm takes as input $S_1, S_2, \dots, S_m$, $R$, and $M$, where $M$ is a sequence of factors obtained from a term in the Liouville equation as described above, and outputs a forest (disjoint union of trees) $T(M)$, which is a graph whose nodes are subsets of $V$. The algorithm proceeds as follows.

\begin{algorithm}[H]
\caption{An algorithm for generating a causal forest}
\begin{algorithmic}
\State $T_0 \gets \{R\}$
\State $M_0 \gets R$
\For{$n \in \{1, \dots, |M|\}$}
\State $T_n \gets  T_{n-1}$
\If{$\exists k<n$ such that $M_n \cap M_k \neq \emptyset$}
\If{$\neg\exists j < n$ such that $M_n \cap M_j = M_n$}
\State $k \gets \min(\{k \text{ s.t. } M_n \cap M_k \neq \emptyset\})$
\State $T_n \gets T_{n-1} \cup (M_n, (M_k, M_n))$ 
\EndIf
\For{$i$ s.t. $S_i \cap M_n \neq \emptyset$ and $S_i \not\in T_{n-1}$ } 
\State $T_n \gets T_n \cup (M_n, (S_i, M_n))$
\EndFor
\Else 
\Return $\emptyset$
\EndIf
\EndFor
\Return $T_{|M|}$
\end{algorithmic}
\end{algorithm}

The result of this algorithm is the forest $T(M)$.
\begin{defn}
If $S_1, \dots, S_m$ are nodes in the forest $T(M)$, then we call $T(M)$ a causal forest. 
\end{defn}
\begin{prop}
If $T(M)$ is not a causal forest then $\prod_{i}\ad_{\mathcal O_{i}} \mathcal L_{M_{|M|}}\mathcal L_{M_{|M|-1}} \cdots \mathcal L_{M_1} |A_R) = 0$.
\end{prop}
\begin{proof}
Let $S_i \notin T(M)$. By the design of the algorithm, this means that $M$ does not have an ordered subsequence of factors connecting $R$ to $S_i$. Since the factors correspond to the support of terms in the Hamiltonian, this means that $|B) \equiv  L_{M_{|M|}}\mathcal L_{M_{|M|-1}} \dots \mathcal L_{M_1}|A_R)$ does not have support on $S_i$, so $\ad_{\mathcal O_{i}}|B) = 0$.
\end{proof}
\begin{defn}
If $T(M)$ is a causal forest, then, by construction, each $S_i$ is contained in a tree 
rooted at $R$, so there is a unique path connecting it to $R$. Excluding the endpoints at $R$ and $S_i$, this path corresponds to a unique sequence $\Gamma_i = (X_1, \dots, X_{l_i})$ of factors $X_i \in F$.
We call $\mathcal S$ the set of equivalence classes of causal forests. An equivalence class of causal forests is defined by pruning all the ``reducible" terms (those not involved in $\Gamma_i$) from $M$, arriving at a minimal ordered sequence $\Lambda$ which will be called the irreducible skeleton.
\end{defn}
An illustration of a causal forest constructed by the algorithm and the corresponding irreducible paths is depicted in Fig.~\ref{fig:causalforest}.

\begin{figure}
\includegraphics[width = .35 \textwidth]{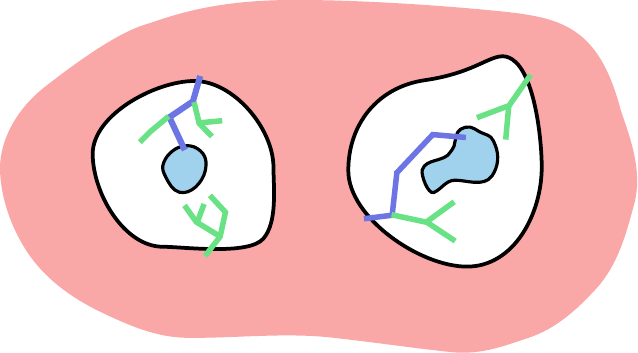}
\caption{Example of a causal forest constructed by the algorithm. The red region again corresponds to $R$ and the two blue regions depict $S_1, S_2$. The blue and green line segments represent factors in the sequence $M$. The blue segments participate in the irreducible paths, and so they form part of the corresponding irreducible skeleton, while the green ones are ``reducible" terms. The ordering of the sequence is not depicted.}
\label{fig:causalforest}
\end{figure}

\subsection{Bounds from irreducible paths}

 The main technical result of this section is to show that under the assumptions outlined above, we can derive a bound on the nested commutator which resembles a product of the single-commutator bounds presented previously. This is captured by the following result:

\begin{thm}
\label{thm:bound}
Consider a local interaction defined by the factor graph $(V,E,F)$. Let $S_1, \dots, S_m$ be sets contained in disjoint regions $B_1, \dots, B_m$ separated by at least one on the subgraph of factors. Consider an operator $A$ with $\Vert A \Vert = 1$ whose support is a set $R$ contained within the complement of $B_1, \dots, B_m$. Let $\mathcal O_1, \cdots \mathcal O_m$ be operators also with norm $1$, where each $\mathcal O_i$ is supported within $S_i$. Then on the nested commutator we have the bound
\begin{equation}
\frac{1}{2^m}\Vert [\mathcal O_m, [\dots,[\mathcal O_1, A(t)]\dots]]\Vert \leq \prod_{i}\sum_{\Gamma \in \Gamma_i(R \to S_i)}\frac{(2t)^{|\Gamma|}}{|\Gamma|!} w(\Gamma)
\label{eq:bound}
\end{equation}
where 
$\Gamma_i(R \to S_i)$ denotes the set of non-self-crossing paths within $B_i$ from $R$ to $S_i$ and $w(\Gamma) = \prod_{X \in \Gamma}\Vert H_{X} \Vert$ is the weight of the path $\Gamma$.
\end{thm}
We begin by reorganizing the terms in Eq.~\eqref{eq:resum} by equivalence class:
\begin{align}
\prod_i\mathbb \ad_{\mathcal O_{i}} e^{\mathcal L t}|A_R) &= 
\prod_i\ad_{\mathcal O_{i}} \sum_{[T] \in \mathcal S} \sum_{M:T(M) \in [T]}\frac{t^{|M|}}{|M|!} \mathcal L_{X_{|M|}}\dots \mathcal L_{X_1}|A_R)
\label{eq:reorganize}
\end{align}
By definition each $T(M) \in [T]$ has the same set of irreducible paths $\{\Gamma_i\}$, so our goal is to rewrite the second sum in Eq.~\eqref{eq:reorganize} by ``factoring out" the contribution from the irreducible paths in the tree, then grouping together and safely re-summing the extraneous terms. This is accomplished in the following lemma.
\begin{lem}\label{lem:causalidentity}
Fix an equivalence class of causal forests $[T]$ with irreducible paths $\{\Gamma_i\}$ and irreducible skeleton $\Lambda$. Then
\begin{align}
\prod_i \ad_{\mathcal O_{i}} \sum_{M:T(M) \in [T]} \frac{t^{|M|}}{|M|!}\mathcal L_{X_n}\cdots \mathcal L_{X_1}|A_R) &= 
\prod_i\ad_{\mathcal O_{i}} \sum_{m_1, m_2, \dots, m_l}\frac{t^n}{n!}(\mathcal L_l)^{m_l}\mathcal L_{\Lambda_l}(\mathcal L_{l-1})^{m_{l-1}} \cdots \mathcal L_{\Lambda_1}(\mathcal L_0)^{m_0}|A_R)
\label{eq:avoidantterms}
\end{align}
where $l \equiv |\Lambda|$ and $n \equiv l + \sum_{i}m_i$. We have defined 
\begin{align}
\mathcal L_i \equiv \sum_{X \cap V_i = \emptyset}\mathcal L_X
\end{align}
such that $V_j$ is the set of ``disallowed" vertices: any factor that intersects $V_j$ would necessarily change the equivalence class of $T(M)$ if it was to appear in $M$ in between $\Lambda_j$ and $\Lambda_{j+1}$. We can break $V_j$ up into the terms disallowed by each path, i.e.
\begin{equation}
V_j = \bigcup_{i}U_{i, \beta_{i}(j)},
\end{equation}
where $\beta_i(j) = \# \{k \leq j : i_k = i\}$ counts the number of terms in sequence $\Lambda$ at or before position $j$ that belong to $\Gamma_i$. Let $l_i = |\Gamma_i|$. Then $U_{i,j}$ is explicitly defined as 
\begin{equation}
U_{i, j} = 
\begin{cases}
\bigcup_{m = \beta_i(j+2)}^{l_i}(\Gamma_i)_{m} & j < l_i - 1, \\
S_i & j = l_i-1>
\end{cases}
\end{equation}
\end{lem}
\begin{proof}
We follow \cite{chen2021operator}.  First, we show that every term on the right-hand side also appears in the summand on the left-hand side of Eq.~\eqref{eq:avoidantterms}. Fix an equivalence class $[T]$ with an irreducible skeleton $\Lambda = (\Lambda_1, \dots, \Lambda_l)$. Let $U_i \equiv \{\mathcal L_X : X \in F, \ X \cap V_i = \emptyset\}$, the set of allowed terms that could appear between $\Lambda_i$ and $\Lambda_{i+1}$, forming another factor sequence $M$ such that $[T(\Lambda)] = [T(M)]$. Then we can write
\begin{align}
&(\mathcal L_l)^{m_l}\mathcal L_{\Lambda_l}(\mathcal L_{l-1})^{m_{l-1}} \cdots \mathcal L_{\Lambda_{1}}(\mathcal L_0)^{m_0}
\notag  \\
&= \qty(\sum_{X_1, \dots, X_{m_l} \in U_l}\mathcal L_{X_1}\cdots \mathcal L_{X_{m_l}}^l)\mathcal L_{\Lambda_l}\qty(\sum_{X_{1}, \dots, X_{m_{l-1}} \in U_{l-1}}\mathcal L_{X_1}\cdots \mathcal L_{X_{m_{l-1}}}) \cdots  \notag\\
&\hspace{.5cm}\times \mathcal L_{\Lambda_{1}} \qty(\sum_{X_1, \dots, X_{m_{0}} \in U_0}\mathcal L_{X_1}\cdots \mathcal L_{X_{m_0}}).
\end{align}
Each term in parentheses is a sum over all possible sequences of allowed terms of length $m_k$, with repeated entries allowed. It is permissible for sequences which do not contribute to a causal tree to appear in the sum because the contribution of these terms must ultimately vanish. By construction, none of the terms $\mathcal L_{X_k}$ intersect any of the irreducible paths $\Gamma_i$ past $\beta_i(k)$, and so they may not create a path connecting $R$ to $S_i$. Therefore, any sequence appearing on the right-hand side which has a non-vanishing contribution still corresponds to a tree $T(M)$ with the irreducible skeleton $\Lambda$. This shows that every non-vanishing term in the summand on the right-hand side corresponds to a causal forest in the equivalence class $[T(\Lambda)]$, and so it appears on the left-hand side. 
\par For the other direction, every $T(M) \in [T(\Lambda)]$ may be constructed as $M_l\Lambda_{l} M_{l-1}\cdots \Lambda_{1}M_0$, where $M_l, \dots, M_0$ are arbitrary sequences of terms that form part of the causal forest, but do not create additional paths between $R$ and $\{S_1, \dots, S_j\}$; otherwise, they would change the equivalence class of $T(M)$. As explained above, these are exactly the sequences that appear in the sum on the right-hand side. Furthermore, each term on the left and each term in the summand on the right correspond to a unique $M$. This establishes a bijection between the non-vanishing terms in the two sums, proving the result.
\end{proof}
\begin{defn}[$n$-simplex]
We define the canonical $n-$simplex with side length $t$, denoted by $\Delta^n(t)$, as an $n-$dimensional shape bounded by $\{x_1 < x_2 < \dots < x_n < t\}$, which generalizes a triangle in two dimensions and a tetrahedron in three dimensions to $n$ dimensions. The volume of $\Delta^n(t)$ is given by
\begin{equation}
\Vol{\Delta^n(t)} = \int\limits_{0}^t \mathrm{d}t_1\int\limits_0^{t_1} \mathrm{d}t_2 \cdots \int\limits_0^{t_{n-1}}\mathrm{d}t_n = \frac{t^n}{n!}
\end{equation}
\end{defn}
In the next step, we can eliminate the contribution of extraneous ``reducible" terms and include in the bound only contributions from the irreducible paths. The engine behind this re-exponentiation is the generalized Schwinger-Karplus identity.
\begin{lem}[Lemma 5 of Ref.~\cite{chen2021operator}]The following identity holds:
\begin{equation}
\sum_{\vec m = \vec 0}^{\infty}\frac{t^n}{n!}\mathcal F_{l}^{m_l}A_l \dots \mathcal F_1^{m_1}A_1 \mathcal F_0^{m_0} = \int\limits_{\Delta^l(t)}\mathrm{d}^l\vec t \mathrm{e}^{\mathcal F_l \Delta t_{l}}A_l \cdots \mathrm{e}^{\mathcal F_{1}\Delta t_1}A_1 \mathrm{e}^{\mathcal F_0 \Delta t_0}
\end{equation}
On the left-hand side we defined the term order $n := l + \sum_{k=1}^{l}m_k$ and on the right-hand side, $\Delta^l(t)$ is the canonical $l$-simplex, and $\Delta t_i = t_{i+1}-t_i$, with $\Delta t_{l} := t- t_l$ and $\Delta t_0 := t_1$.
\label{lem:SchwingerKarplus}
\end{lem}

\begin{proof}[Proof of Theorem \ref{thm:bound}]
By applying Lemmas~\ref{lem:causalidentity} and ~\ref{lem:SchwingerKarplus} to Eq.~\eqref{eq:avoidantterms}, we obtain
\begin{align}
\left\Vert\frac{1}{2^m}\prod_i\ad_{\mathcal O_{i}} \mathrm{e}^{\mathcal L t}|A_R)\right\Vert
&\leq  \frac{1}{2^m}\sum_{[T] \in \mathcal S} \left\Vert\prod_i\ad_{\mathcal O_{i}} \sum_{T(M) \in [T]}\frac{t^{|M|}}{|M|!} \mathcal L_{X_{|M|}}\dots \mathcal L_{X_1}|A_R)\right\Vert \notag \\
& \leq \sum_{\Lambda \in \mathcal S}\left\Vert \sum_{m_1, m_2, \dots, m_q}\frac{t^n}{n!}(\mathcal L_l)^{m_l}\mathcal L_{\Lambda_{l}}(\mathcal L_{l-1})^{m_{l-1}} \dots \mathcal L_{\Lambda_{1}}(\mathcal L_0)^{m_0}|A_R)\right\Vert \notag \\
&\leq  \sum_{\Lambda \in \mathcal S}
\int\limits_{\Delta^{l}(t)}\mathrm{d}^{l} \vec t\left\Vert \exp(\mathcal L_l\Delta t_{l})\mathcal L_{\Lambda_{l}}\exp(\mathcal L_{l-1}\Delta t_{l-1}) \dots \mathcal L_{\Lambda_{1}}\exp(\mathcal L_0\Delta t_0)|A_R)
\right\Vert \notag \\
&\leq \sum_{\Lambda \in \mathcal S}\Vol(\Delta^{l}(t))2^{l}\prod_{i = 1}^{l}\Vert H_{\Lambda_{i}}\Vert
\label{eq:paths}
\end{align}
In the second line we use the fact that $\Vert \ad_{\mathcal O_{i}}\Vert_{op} \leq 2$ for $\Vert \mathcal O_{i} \Vert = 1$, and in the fourth line we use the unitary invariance and submultiplicativity of the operator norm. Since $l = \sum_{j}|\Gamma_j|$ by definition, the canonical $n-$simplices have the property that
\begin{equation}
\Vol(\Delta^{l}(t)) = \frac{1}{l!}\prod_{j}|\Gamma_j|!\Vol(\Delta^{|\Gamma_j|}(t)) \ ,
\end{equation}
so we can rewrite the right-hand side of \eqref{eq:paths} as 
\begin{align}
 \sum_{\Lambda \in \mathcal S}\Vol(\Delta^{l}(t))2^{l}\prod_{i = 1}^{l}\Vert H_{\Lambda_i}\Vert 
 &= \sum_{\{\Gamma_j\}: (\Gamma_1, \dots, \Gamma_m) \in \mathcal S}\frac{1}{l!}\prod_{j = 1}^m|\Gamma_j|!\Vol(\Delta^{|\Gamma_j|}(t))2^{|\Gamma_j|} w(\Gamma_j)
\end{align}
where the sum on the right-hand side is taken over all possible sets of irreducible paths $\{\Gamma_j\}$ such that the irreducible skeleton $\Lambda = (\Gamma_1, \dots, \Gamma_m)$ obtained from their concatenation belongs to $\mathcal S$. The summand does not depend on the ordering of the individual paths, only on the set of irreducible paths $\{\Gamma_j\}$ so we can reorganize the sum over sets of paths $\{\Gamma_j\}$:
\begin{align}
\label{eq:productbound}
\sum_{\Lambda \in \mathcal S}\Vol(\Delta^{l}(t))2^{l}\prod_{i = 1}^{l}\Vert H_{\Lambda_i}\Vert  &= \sum_{\{\Gamma_i\}}\frac{l!}{|\Gamma_1|!\cdots |\Gamma_m|!}\frac{1}{l!}\prod_{i = 1}^m|\Gamma_i|!\Vol(\Delta^{|\Gamma_i|}(t))2^{|\Gamma_i|} w(\Gamma_i) \notag \\
&=\sum_{\{\Gamma_i\}}\prod_{i = 1}^m\Vol(\Delta^{|\Gamma_i|}(t))2^{|\Gamma_i|} w(\Gamma_i) = \prod_{i}\sum_{\Gamma \in \Gamma_i(R \to S_i)}\frac{(2t)^{|\Gamma|}}{|\Gamma|!} w(\Gamma).
\end{align}
Hence we obtain Eq.~\eqref{eq:bound}. \end{proof}

\subsection{Combinatorial bounds}
Although Eq.~\eqref{eq:bound} is formally quite elegant, it is not necessarily always convenient to work with.  In this subsection, we obtain more physically transparent nested commutator bounds by further simplifying Eq.~\eqref{eq:bound}. First, the nested commutator can be bounded using the equivalence-class bounds on single commutators from Theorem~\ref{thm:eqconstruction}.
\begin{cor}[Decoupling corollary]
\label{cor:decoupling corollary}
Let $C_{ij}(t)$ be the bound on the single commutator in Theorem~\ref{thm:eqconstruction}. Then 
\begin{equation}
\frac{1}{2^m}\Vert [\mathcal O_m,[\dots,[\mathcal O_1, A(t)]\dots] ]\Vert
\leq  
\prod_{i}\sum_{\substack{u \in \partial B_{i} \\ v \in \partial S_i}}C_{uv}(t)
\end{equation}
\end{cor}
This bound is not quite as tight as Theorem~\ref{thm:bound} because we now include paths that enter the region $R$. 
At short times, we find that the following result holds.
\begin{cor}
Suppose that $\Vert H_X \Vert \le h$ for all $X$ that intersect with $B_1, \dots, B_m$. Let $r_i := \mathsf{d}(R, S_i)$, where distance is defined via the factor graph as in Def.~\ref{def:distance}, and $\Delta$ be the maximum degree of the factor graph. If 
\begin{equation}
    |t| < \frac{\min_i(r_i)}{2h\Delta}
\label{eq:cor48t}
\end{equation} 
then 
\begin{equation}
C_{\vec S}^{R}(t) \leq \prod_i |\partial B_i||\partial S_i|\qty(\frac{2\mathrm{e}h|t|\Delta}{r_i})^{r_i},
\end{equation}
where $C_{\vec S}^{R}(t)$ is defined as in Eq.~\eqref{eq:nest_comm_def}.
\label{cor:combinatorialbound}
\end{cor}
\begin{proof}
The number of distinct non-self-crossing paths of length $l$ from $u \in \partial B_i$ to $v_i$ must be smaller than $\Delta^l$. Given a ball $B_i$, the minimum length of path $\Gamma_i$ in a causal tree is $r_i \equiv d(R, v_i)$. Therefore we have
\begin{equation}
\sum_{\Gamma_i(R \to v_i)}\frac{(2t)^{|\Gamma|}}{|\Gamma|!} w(\Gamma) \leq  \sum_{l = r_i}^{\Vol(B_i)}\frac{(2th\Delta)^{l}}{l!}.
\end{equation}
For any $\alpha > 1$, $\alpha^{l-r_i} > 1$ whenever $l > r_i$.  We therefore find that 
\begin{equation}
\sum_{l = r_i}^{\Vol(B_i)}\frac{(2th\Delta)^{l}}{l!} \leq 
\inf_{\alpha > 1}\sum_{l = 0}^{\infty}\frac{(2th\Delta)^{l}}{l!}\alpha^{l-r_i} = 
\inf_{\alpha > 1}\alpha^{-r_i}e^{2th\alpha\Delta}.
\label{eq:cor48step}
\end{equation}
If Eq.~\eqref{eq:cor48t} holds, plugging in $\alpha = (2h\Delta t)^{-1}r_i > 1$ into Eq.~\eqref{eq:cor48step} gives
\begin{equation}
\inf_{\alpha > 1}\alpha^{-r_i}\mathrm{e}^{2th\alpha\Delta} \leq 
\qty(\frac{2\mathrm{e} h\Delta |t|}{r_i})^{r_i}.
\end{equation}
Inserting this bound into Eq.~\eqref{eq:bound}, we have
\begin{equation}
C_{\vec S}^{R}(t) \leq \prod_i |\partial B_i||\partial S_i|\qty(\frac{2\mathrm{e} h\Delta|t|}{r_i})^{r_i}
\end{equation}
as desired.
\end{proof}

We can also encode the combinatorial problem of summing weighted paths between two points on the factor graph into a matrix exponential, which can be computationally useful.
\begin{cor}
Define a real-symmetric matrix $h$ elementwise by
\begin{equation}
h_{ij} = 
\begin{cases}
\sum_{ X \in F: i,j \in X}\Vert H_X\Vert & i \neq j \\
0 & i = j
\end{cases}
\end{equation}
Then
\begin{equation}
C_{\vec S}^R(t) \leq \prod_i \sum_{\substack{u \in \partial B_i \\ v \in \partial S_i}}\exp(2h|t|)_{uv}
\end{equation}
\end{cor}
\begin{proof}
The result follows immediately by combining Cor.~\ref{cor:exp_path_counting} with Cor.~\ref{cor:decoupling corollary}.
\end{proof}

Perhaps most importantly, we would like to have a bound that makes the \emph{volume-law decay} of large nested commutators extremely transparent.  This is accomplished with the following result.
\begin{cor}
Let $B_R$ denote a metric ball of radius $R$ with respect to the geodesic distance on the factor graph (as defined in Def.~\ref{def:distance}). There exist constants $\gamma, c_{\rm{LR}} > 0$ such that we can choose $v_1, \dots, v_m \in B_R$ (where $m$ depends on $R$) for which
\begin{align}
C_{v_1, \dots, v_m}^{B_R^c}(t)  \leq c_{\rm{LR}}\exp(-\gamma\frac{(R-\LRvel t)^d}{(\LRvel t)^{d-1}})
\end{align}
for any $\LRvel t > 1$.
\label{cor:volumescaling}
\end{cor}
\begin{proof}
We choose $S_1, \dots, S_m$ to be vertices $v_1, \dots, v_m$ and $B_1, \dots, B_m$ to be metric balls $B_i = B_\xi(v_i)$ of radius $\xi$ centered on these vertices. We will optimize $\xi$ to obtain (\ref{cor:volumescaling}). First suppose that $d > 1$. For a lattice in $d$ dimensions, there exists a constant $C_1$ such that $|\partial B_\xi| \leq (C_1\xi)^{d-1}$. 
Then directly applying Cor.~\ref{cor:combinatorialbound}, we have
\begin{align}
C_{\vec v}^R(t) \leq \qty[(C_1\xi)^{d-1}\qty(\frac{\LRvel t}{\xi})^{\xi}]^m.
\end{align}
Let $\xi = \alpha v_{\text{LR}}t$ for some $\alpha > 1$. The requirement that ensures that $\alpha > 1$ ensures that we can apply Cor.~\ref{cor:combinatorialbound}. There exists a constant $C_2$ and, for any $\alpha$, a constant $K$, such that for all $\LRvel t > 1$ and all $R  > K\LRvel t$, we can find $m$ distinct points $v_1, \ldots v_m$ in $B_R$ separated pairwise by more than $2\xi$, where $m > C_2\qty(\frac{R}{\xi})^d$.
 The bound above may then be written as
\begin{align}
C_{\vec v}^{B_R^c} &\leq 
\qty[\qty(C_1\alpha \LRvel t)^{(d-1)\alpha^{-d} (\LRvel t)^{-1}}\qty(\frac{1}{\alpha})^{\alpha^{1-d}}]^{C_2 R^d(\LRvel t)^{1-d}} \notag \\
&\leq \qty[\qty(\frac{\mathrm{e}^{C_1 (d-1)}}{\alpha})^{\alpha^{1-d}}]^{C_2 R^d(\LRvel t)^{1-d}}
\label{eq:alphabound}
\end{align}
Then for any $\alpha > \exp(C_1(d-1))$, we have $(\mathrm{e}^{C_1(d-1)}\alpha^{-1})^{\alpha^{d-1}} \equiv \mathrm{e}^{-\beta} < 1$. Plugging this in gives us
\begin{align}
C_{\vec v}^{D^c}(t) \leq \exp(-\beta C_2 R^d/(\LRvel t)^{d-1})
\end{align}
Additionally, $\beta$ achieves a maximum at $\alpha = \exp(C_1(d-1)+\frac{1}{d-1}) > 1$ of $\beta = \exp(-C_1(d-1)^2-1)/(d-1)$. 
This holds for $R > K\LRvel t$, but we would like a bound that holds for all $R > \LRvel t$. 
If $\frac{R}{\LRvel t} < K$ then we can apply the conventional Lieb-Robinson bound (Thm.~\ref{thm:standardLR}). Choosing $v_1$ at the center of $B_R$, we have
\begin{align}
C_{\vec v}^{B_R^c} \leq \frac{1}{2}\Vert [\mathcal O_{v_1}, A(t)]\Vert \leq c_{\text{LR}}\mathrm{e}^{\mu (\LRvel t - R)}
\end{align}
Let $C_2' \equiv \beta C_2$, and, for simplicity, assume that $c_{\text{LR}} \geq 1$. For any $\lambda \leq 1$, we clearly have
\begin{align}
\exp(-C_2' R^d/(\LRvel t)^{d-1}) \leq c_{\text{LR}}\exp(-C_2'\lambda(R-\LRvel t)^{d}/(\LRvel t)^{d-1})
\end{align}
Then we find
\begin{align}
e^{\mu (\LRvel t - R)} \leq \exp(-C_2'\lambda(R-\LRvel t)^{d}/(\LRvel t)^{d-1})
\end{align}
when $\lambda \leq \mu/[C_2'(\frac{R}{\LRvel t}-1)^{d-1}]$. If $1<R/\LRvel t < K$ then $\lambda = \min(\mu/[C_2'(K-1)^{d-1}],1) < \mu/[C_2'(\frac{R}{\LRvel t}-1)^{d-1}]$, so, for all $\frac{R}{\LRvel t} > 1$, we have
\begin{align}
C_{\vec v}^{B_R^c}(t) \leq \exp(-\lambda C_2' (R-\LRvel t)^d/(\LRvel t)^{d-1}).
\end{align}
This completes the proof.
\end{proof}

\subsection{Generalizing to time-dependent interactions}
\label{sec:time-dependent}
Lieb-Robinson bounds, at least those developed here, are ``worst-case-scenario" bounds. This is because the operator norm measures the \textit{most} that an operator can increase the norm of a vector, and then we repeatedly apply the triangle inequality. Thus, the argument that leads to the bounds is fundamentally combinatorial in nature, and it should be unsurprising that these arguments generalize to time-dependent systems. In this subsection, we explain why the argument in Sec.~\ref{sec:nestedcommutatorbounds} generalizes to time-dependent systems.

In the time-dependent case, we can instead expand the propagator as a Dyson series. Let $\mathcal L(t) \equiv -\ii \ad_{H(t)}$ be the system Liouvillian. Furthermore, we write $\mathcal L_{X}(t) \equiv - \ii\ad_{H_X(t)}$ for $X \subset V$ when the support of $H_X(t)$ is contained within $X$ at all times. The Dyson series expansion is
\begin{align}
\e^{\mathcal L t} = \sum_{X_1, \dots, X_n \subset V}\int_{0}^{t}\dd s_n \int_{0}^{s_n}\dd s_{n-1} \cdots \int_{0}^{s_2}\dd s_1\mathcal L_{X_n}(s_n) \cdots \mathcal L_{X_1}(s_1) \equiv \sum_{X_1, \dots, X_n \subset V}\mathcal I^{(n)}_t(\mathcal L_{X_1}, \dots, \mathcal L_{X_n})\ ,
\end{align}
where $n$ is summed over and we have defined $\mathcal I^{(n)}_t$ as the convolution operator above labeling sequences of operators appearing in the expansion. This is the appropriate generalization of the Taylor-series expansion in the time-independent case \eqref{eq:resum}. Despite the time-dependence, the terms in Hamiltonian $H_X(t)$ have supports which are bounded by $X$, so the algorithm for constructing the equivalence classes is unaffected. Therefore, we simply need to develop a suitable generalization of the Schwinger-Karplus identity (Lemma~\ref{lem:SchwingerKarplus}).

\begin{lem}
Given time-dependent operators $\mathcal F_l, \mathcal A_l, \dots, \mathcal F_1, \mathcal A_1, \mathcal F_0$, the following identity holds:
\begin{align}
&\sum_{m_0, \dots, m_l = 0}^{\infty}\mathcal I_t^{(l+\sum m_l)}(\mathcal F_l\stackrel{\times m_l}{\dots}, \mathcal A_l, \dots, \mathcal F_1\stackrel{\times m_1}{\dots},\mathcal A_1, \mathcal F_0\stackrel{\times m_0}{\dots}) \notag \\
&=\int_0^{t} \dd s_l \int_{0}^{s_l}\dd s_{l-1} \cdots \int_{0}^{s_2}\dd s_1  \mathcal U^{(l)}(t, s_l) \mathcal{A}_l(s_l) \mathcal U^{(l-1)}(s_l, s_{l-1})\mathcal{A}_{l-1}(s_{l-1}) \cdots \notag\\
&\hspace{2cm}\times \mathcal U^{(1)}(s_2, s_1)\mathcal A_1(s_1)\mathcal U^{(0)}(s_1, 0).
\end{align}
Here
\begin{equation}
\mathcal U^{(k)}(t, t_0) \equiv \mathcal T\exp(\int_{t_0}^t \dd s\mathcal F_k(s))
\end{equation}
with $\mathcal T$ standing for time-ordering.
\end{lem}
In the time-dependent form, the expansion is arguably more transparent than the time-independent form. For convenience, given any operator $\mathcal F$, define
\begin{align}
\mathcal F^{(n)}(t, t_0) \equiv \int_{t_0}^{t} \dd s_n \int_{t_0}^{s_n}\dd s_{n-1} \cdots \int_{t_0}^{s_2}\dd s_1 \mathcal F(s_n) \cdots \mathcal F(s_1),
\end{align}
where $t_0 = 0$ is implied if the second argument is omitted.
\begin{proof}
Consider the Dyson series expansion
\begin{equation}
\mathcal U^{(k)}(t,t_0) = \mathcal T \exp(\int_{t_0}^{t} \dd s \mathcal F_k (s)) = \sum_{n= 0}^{\infty} \mathcal F_k^{(n)}(t,t_0).
\end{equation}
Examining the left-hand side of the equation above, we can play with the time-ordering to put it in a form resembling this:
\begin{align}
&\sum_{m_0, \dots, m_l = 0}^{\infty}\mathcal I_t^{(l+\sum m_l)}(\mathcal F_l \stackrel{\times m_l}{\dots},\mathcal A_l, \dots, \mathcal F_1 \stackrel{\times m_1}{\dots}, \mathcal A_1, \mathcal F_0\stackrel{\times m_0}{\dots}) \notag \\
&=\int_{0}^t \dd s_l  \cdots \int_{0}^{s_2}\mathrm{d}s_1\qty(\sum_{m_l=0}^\infty\mathcal F_l^{(m_l)}(t, s_l)) \mathcal A_l(s_l) \cdots \qty(\sum_{m_1 = 0}^\infty\mathcal F_{1}^{(m_1)}(s_2, s_1))\mathcal A_{1}(s_1)\qty(\sum_{m_0 = 0}^\infty\mathcal F_{0}^{(m_0)}(s_1, 0)) \notag \\
&=\int_{0}^t \dd s_l  \cdots \int_{0}^{s_2}\mathrm{d}s_1\mathcal U^{(l)}(t, s_l)\mathcal A_l(s_l) \cdots\mathcal U^{(1)}(s_2, s_1)\mathcal A_{1}(s_1)\mathcal U^{(0)}(s_1, 0).
\end{align}
The second line follows by repeatedly applying 
\begin{equation}
\int_{0}^{s_{k}}\dd s_k'\int_0^{s_k'} \dd s_{k-1} \int_{0}^{s_{k-1}}\dd s_{k-1}'[\dots] = \int_{0}^{s_k} \dd s_k' \int_{0}^{s_{k}'}\dd s_{k-1}' \int_{s_{k-1}'}^{s_k'} \dd s_k[\dots]
\end{equation}
to interchange the order of the integrals.
\end{proof}
From here, one recognizes that the manipulations proceed exactly as in the time-independent case. Fix an equivalence class of causal forests $[T]$ with an irreducible skeleton $\Lambda$. Then
\begin{align}
&\prod_i \ad_{\mathcal O_i} \sum_{M:T(M) \in [T]} \int_{0}^t \dd s_n \cdots \int_{0}^{s_2}\dd s_1 \mathcal L_{X_n}(s_l) \cdots \mathcal L_{X_1}(s_1)|A_R) \notag \\
&= \prod _i \ad_{\mathcal O_i}\sum_{\vec m} \int_{0}^t \dd s_l \cdots \int_{0}^{s_2}\dd s_1 \mathcal L_{l}^{(m_l)}(t,s_l)\mathcal L_{\Lambda_l}(s_l) \cdots \mathcal L_{1}^{(m_1)}(s_2, s_1)\mathcal L_{\Lambda_1}(s_1)\mathcal L_0^{(m_0)}(s_1)|A_R) \notag \\
&= \prod _i \ad_{\mathcal O_i}\int_{0}^t \dd s_l \cdots \int_{0}^{s_2}\dd s_1 \mathcal U^{(l)}(t,s_l)\mathcal L_{\Lambda_l}(s_l) \cdots \mathcal U^{(1)}(s_2, s_1)\mathcal L_{\Lambda_1}(s_1)\mathcal U^{(0)}(s_1)|A_R) 
\end{align}
Then taking norms, we have
\begin{align}
&\left\Vert\prod_{i=1}^m \ad_{\mathcal O_i} \sum_{M:T(M) \in [T]} \int_{0}^t \dd s_l \cdots \int_{0}^{s_2}\dd s_1 \mathcal L_{X_n}(s_l) \cdots \mathcal L_{X_1}(s_1)|A_R)\right \Vert  \notag \\
&\leq  2^m\int_{0}^t \dd s_l \cdots \int_{0}^{s_2}\dd s_1 \Vert \mathcal L_{\Lambda_l}(s_l)\Vert \cdots \Vert\mathcal L_{\Lambda_1}(s_1)\Vert \notag \\
&\leq 2^m\frac{(2t)^l}{l!}\prod_{X \in \Lambda}\sup_{s \leq t}\Vert H_{X}(s)\Vert
\end{align}
which appropriately generalizes Eq.~\eqref{eq:paths}.

\section{Hamiltonians with exponential tails}\label{sec:exponential}
For some of our applications, it is important to extend our bounds to Hamiltonians with \emph{exponentially decaying} interactions. The main issue is that we are no longer able to write each causal tree as a disjoint union of irreducible paths between region $R$ and each $S_i$. In this section, we describe how to overcome this difficulty to generalize many of the results of the previous section to interactions with exponential tails. These results will also apply to local Hamiltonians, although they are sometimes weaker than the strongest bounds in the case of local interactions, and the proof methods are significantly more nuanced.\footnote{This kind of result also occurs for the usual single commutator Lieb-Robinson bounds, where the upper bounds from quasi-local Hamiltonian evolution also applies to local Hamiltonian evolution, but techniques specialized to local Hamiltonians can give a bound with slightly faster decay outside of the light cone \cite{AnthonyChen:2023bbe}.}

\subsection{Review of Lieb-Robinson bounds for quasi-local Hamiltonians\label{sec:single_commutator_exp_bound}}

In the previous section regarding strictly local interactions, we defined locality via the factor graph. With the introduction of exponentially decaying interactions, this definition is no longer appropriate, because it is the strength of the interactions rather than their support which defines locality. We can still apply the irreducible path bounds from the previous section, but the degree of the factor graph is unbounded, and more care must be taken to correctly sum up the weights of all the potential paths. 

In this section, our setting will be qudits on the sites of a graph $G = (V, E)$ with finite degree $\Delta$. Unlike in the previous sections, we use the distance $\mathsf d(x,y)$ to refer to the geodesic distance provided by this underlying graph $G$. The volume of connected sets in this graph will then constrain the weights of edges on the factor graph.

\begin{defn}\label{defquasilocal}
We say that a Hamiltonian $H$ is quasi-local if 
\begin{equation}
H = \sum_{S\subseteq V}H_S  \text{ such that $S$ is connected and } \Vert H_S \Vert \leq h\mathrm{e}^{-\kappa |S|} \label{eq:HS_decomposition}
\end{equation}
for some constants $h$ and $\kappa$. A subset of the graph is connected if for each $u \in S$ there is some $v \in S$ such that $\mathsf d(u, v) = 1$. Here $H_S$ does not need to act non-trivially on every site in $S$, but decomposition \eqref{eq:HS_decomposition} may only exist if most terms in $H$ are associated with the smallest possible set $S$. In this paper, for convenience, we restrict to 
\begin{equation}
   \kappa > \log(\Delta)+1. \label{eq:kappalowerbound}
\end{equation}
\end{defn}

The aim of this section is to generalize the following conventional Lieb-Robinson bound in the same way we generalized bounds for local systems in the previous section.

\begin{prop}[Theorem 3.7 of Ref.~\cite{AnthonyChen:2023bbe}]
\label{prop:quasilocalbound}
Let $H$ be quasi-local. Suppose that $O$ is supported on a set $S \subset B$ and $A$ is supported on a set $R$ which is disjoint from but adjacent to $B$ as in Fig.~\ref{fig:setup}. Then there exist constants $c_{\rm LR}, \mu, \LRvel$ determined by $h,\kappa,\Delta$ such that 
\begin{align}
\frac{1}{2}\Vert [O, A(t)]\Vert \leq c_{\rm LR}|\partial B||\partial S|\mathrm{e}^{-\mu\mathsf d(R, S)}\left(\mathrm{e}^{\mu \LRvel t} - 1\right).
\end{align}
\end{prop}

In the case of strictly local Hamiltonians, a path needed a minimum length, measured on the factor graph, to couple two regions. Since we can no longer do this for quasi-local Hamiltonians, we have to use a different strategy.

Conventional LR bounds follow from a slightly less restrictive condition than our criterion for quasi-locality \cite{AnthonyChen:2023bbe}. A more standard definition of quasi-local in the literature is to require that there exist $h, \mu$ such that, for any two sites $u,v \in V$, we have $\sum_{X \ni u,v} \Vert H_X \Vert \leq h\e^{-\mu \mathsf d(u,v)}$. Intuitively, it requires that Hamiltonian terms decay exponentially in their \emph{radius} rather than their entire \emph{support/volume}. We can show that this weaker condition follows readily from our definition using the following proposition.

\begin{prop}[Adapted from Ref.~\cite{highT_Haah}]\label{prop:clustercounting}
    Given a vertex $u \in V$ on a graph of maximal degree $\Delta$, there are at most $(\Delta \mathrm{e})^m$ connected subsets $S$ with $u \in S$ and $|S| = m$. Furthermore, these subsets can be enumerated classically in runtime $m\Delta(\Delta \mathrm{e})^m$.
\end{prop}
This is obtained directly by Proposition 4.6 and Sec.~4.3 of Ref.~\cite{highT_Haah}. 

\begin{lem}
\label{lem:factorsum}
Let $u,v \in V$ be arbitrary. Then, for any $\alpha > 0$, there exist constants $\mu$, $h'$ such that the following holds:
\begin{align}
    \sum_{X \ni u,v}\Vert H_X \Vert \leq \frac{h'\mathrm{e}^{-\mu\mathsf d(u,v)}}{\mathsf d(u,v)^\alpha}.
\end{align}
\end{lem}
\begin{proof}
We begin by overbounding the sum with every set which includes $u$ and has at least $\mathsf d(u, v)$ members, as required of a connected set containing both $u$ and $v$:
\begin{align}
\sum_{X \ni u,v}\Vert H_X \Vert \leq \sum_{m \geq \mathsf d(u,v)}\sum_{\substack{X \ni u \\ |X| =m}}\Vert H_X \Vert
\leq h\sum_{m \geq \mathsf d(u,v)}(\e\Delta)^m\mathrm{e}^{-\kappa m}  
= \frac{h \mathrm{e}^{-[\kappa+\log(\e\Delta)]\mathsf d(u,v)}}{1-(\e\Delta)} \mathrm{e}^{-\kappa} =: C \mathrm{e}^{-\mu'\mathsf d(u,v)}. \label{eq:usualquasilocal}
\end{align}
We have bounded the number of connected subsets $X$ containing $u$ of size $m$ with $(e\Delta)^m$ by Prop.~\ref{prop:clustercounting}. Then the result is a geometric series, which converges due to Eq.~\eqref{eq:kappalowerbound}. Now let $\alpha > 0$. Since an exponential decays much more rapidly than a power-law, we can find $\mu$ dependent on $\alpha$ such that
\begin{align}
C \mathrm{e}^{-\mu'\mathsf d(u,v)} \leq
h' \frac{\mathrm{e}^{-\mu\mathsf d(u,v)}}{\mathsf d(u,v)^\alpha}.
\end{align}
where $h'$ is another constant dependent on $\alpha$ and $C$, but not dependent on $\mathsf d(u,v)$. This is the advertised result.
\end{proof}
The reason for requiring the power-law in the denominator is that we need a way to iteratively apply this bound in order to constrain the weight of irreducible paths. This is called the \emph{reproducing} property~\cite{hastings2010quasi}, and is encapsulated by the following lemma.
\begin{lem}
For any $\alpha > d$, there exists a constant $K$ independent of $u,v$ such that
\label{lem:sumoverpaths}
\begin{equation}
\sum_{k \neq u, v}\frac{\mathrm{e}^{-\mu\mathsf d(u,k)}}{\mathsf d(u,k)^\alpha}\frac{\mathrm{e}^{-\mu\mathsf d(k,v)}}{\mathsf d(k,v)^\alpha} \leq K \frac{\mathrm{e}^{-\mu \mathsf d(u,v)}}{\mathsf d(u,v)^\alpha}
\end{equation}
\end{lem}
This makes the function $G_\alpha(l) \equiv \mathrm{e}^{-\mu l}l^{-\alpha}$ with $\alpha > d$ reproducing for the given lattice. Bounding the weight of edges in the factor graph by a reproducing function of distance on the underlying graph is exactly the criterion which allows us to easily bound the sum of weights of irreducible paths through the factor graph by a function of the distance between the start and end points. We will illustrate how this property is deployed to constrain the locality of dynamics by reproducing the proof of Prop.~\ref{prop:quasilocalbound} found in Ref.~\cite{AnthonyChen:2023bbe}.
\begin{proof}[Proof of Prop.~\ref{prop:quasilocalbound}]
Applying the result from Theorem~\ref{thm:chen_LR_bound}, the commutator is bounded by 
\begin{align}
\frac{1}{2}\Vert [O, A(t)]\Vert  \le \sum_{m = 1}^{\infty}\frac{(2t)^m}{m!}\sum_{u \in \partial B, v \in \partial S}\sum_{\substack{\Gamma(u \to v) \\ |\Gamma| = m}} w(\Gamma) 
\label{eq:exponentialpaths}
\end{align}
Now we need to bound the sum over weighted subsets. First fix some $u \in \partial B, v \in \partial S$. Then consider the irreducible path $\Gamma = (X_1, X_2, \dots, X_m)$. We note that the sets $X_i$ must form a connected path, i.e. $X_k \cap X_{k-1} \neq \emptyset$, so we can overbound the sum by summing over all such connected multisets. Furthermore, we can pick one point $v_k$ in each intersection $X_{k} \cap X_{k-1}$. Since this is true for each such connected multiset, we can further overbound the sum by considering all such sequences of vertices $\vec v = (u, v_2, \dots, v_{m-1}, v)$, where $v_2, \dots, v_{m-1}$ are arbitrary, and sum over $\vec X = X_1, \dots, X_m$ satisfying $v_k \in X_k \cap X_{k-1}$: 
\begin{align}
\sum_{\substack{\Gamma(u \to v) \\ |\Gamma| = m}}w(\Gamma) &\leq \sum_{v_2, \dots, v_{m-1}}\sum_{X_1 \ni u,v_2}\Vert H_{X_1}\Vert \cdots \sum_{X_m \ni v_{m-1}, v}\Vert H_{X_m} \Vert \notag\\
&\leq h'^m\sum_{\vec v}\frac{\mathrm{e}^{-\mu\mathsf d(u,v_2)}}{\mathsf d(u, v_2)^\alpha} \cdots \frac{\mathrm{e}^{-\mu\mathsf d(v_{m-1},v)}}{\mathsf d(v_{m-1}, v)^\alpha} \notag\\
&\leq \frac{(h'K)^m}{K}\frac{\mathrm{e}^{-\mu\mathsf d(u, v)}}{\mathsf d(u, v)^\alpha}.
\label{eq:pathcounting}
\end{align}
Thus we have the bound
\begin{align}
\sum_{u \in \partial R, v \in \partial S}\sum_{\substack{\Gamma(u \to v) \\ |\Gamma| = m}} w(\Gamma) 
 \leq |\partial B| |\partial S|\frac{(h'K)^m}{K}\frac{\mathrm{e}^{-\mu\mathsf d(R, S)}}{\mathsf d(R, S)^\alpha},
 \label{eq:single_path_bound}
\end{align}
Inserting this into Eq.~\eqref{eq:exponentialpaths}, we find that
\begin{align}
C_S^R(t) &\leq |\partial B| |\partial S|\frac{\mathrm{e}^{-\mu\mathsf d(u, v)}}{K\mathsf d(u, v)^\alpha}\sum_{m = 1}\frac{(2t)^m}{m!}(h'K)^m \notag\\
&= |\partial B| |\partial S|\frac{\mathrm{e}^{-\mu\mathsf d(u, v)}}{K\mathsf d(u, v)^\alpha}\left[\exp(2t Kh')-1\right] = c |\partial B| |\partial S|\mathrm{e}^{-\mu \mathsf d(R,S)}(\mathrm{e}^{\mu \LRvel t}-1)
\end{align}
where we have introduced the constants $c = K^{-1}$ and $\LRvel = 2Kh'\mu^{-1}$.
\end{proof}

\subsection{Extending to nested commutators: $m = 2$ case}
Now we can extend this result to nested commutators.
For an illustrative example, we first consider the case where there are only two regions $B_1, B_2$.
In this section, we illustrate the strategy for the general case by proving the following proposition.
\begin{prop}
\label{prop:double_commutator_bound} For operators $\mathcal{O}_1,\mathcal{O}_2$ supported on $S_1,S_2$, respectively, and $A$ supported in the complement of $B_1\cup B_2$, with $\mathcal{O}(1)$ constants $\mu,v_{\mathrm{LR}}, h, h'$ as defined in the previous section, we have
\begin{equation}
\frac{1}{4}\Vert[\mathcal O_1, [\mathcal O_2, A(t)]] \Vert \leq \mu \LRvel t (\mu \LRvel t + e^{-2(\mu \chi-\kappa)}) \prod_{i =1,2}b|\partial B_i||\partial S_i|e^{\mu (\LRvel t -\mathsf d(S_i, B_i))}
\end{equation}
where $b \equiv \frac{h'}{h}$ and $\chi \leq \mathsf d(B_1, B_2)/2$.
\end{prop}
Before proceeding with the proof, we introduce a few definitions. The most important distinction with the single-region case in the previous section is that we now have factors which intersect both $B_1$ and $B_2$, so we can no longer immediately factor the paths into products within each $B_i$, as in Corollary~\ref{cor:decoupling corollary}. However, we can reduce to the decoupled case by sorting the equivalence classes in the following way.
\begin{defn}
Let $[\Gamma] \in \mathcal S$ be an equivalence class with irreducible paths $\Gamma_1$ and $\Gamma_2$ connecting $S_1$ and $S_2$ to $R$ respectively. We will say that $\Gamma$ couples $B_1$ and $B_2$ if $\Gamma_1 \cap \Gamma_2 \neq \emptyset$. 
\end{defn}
\begin{lem}
If $\Gamma$ couples $B_1$ with $B_2$, then $\Gamma_1 \cap \Gamma_2 = \{X\}$, where $X$ is a factor. $X$ is the first element of both $\Gamma_1$ and $\Gamma_2$, and the only element of either path that intersects $R$. 
\label{lem:pathbreaking}
\end{lem}
\begin{proof}
Let $X \in \Gamma_1 \cap \Gamma_2$. The first element of an irreducible path is the only element of said irreducible path that can intersect with $R$ by construction. Since $X$ is required to be connected, it must intersect with $R$. Therefore it must be the first element in the irreducible path $\Gamma_1$ and similarly for $\Gamma_2$. If $Y \in \Gamma_1 \cap \Gamma_2$, then $X$ and $Y$ are both the first elements of both paths, so $X = Y$. 
\end{proof}
This shows in particular that the irreducible skeletons in the coupled case can be reduced to the uncoupled case by removing the first term.
\begin{proof}[Proof of Prop.~\ref{prop:double_commutator_bound}]
We will now divide the equivalence classes into uncoupled paths $\mathcal S_d$ and coupled paths $\mathcal S_c$. Then the commutator bound breaks into a sum:
\begin{align}\label{eq:breaksum}
\frac{1}{4}\Vert[\mathcal O_1, [\mathcal O_2, A(t)]] \Vert \leq \sum_{\Gamma \in \mathcal S_d}\frac{(2t)^{|\Gamma|}}{|\Gamma|!}w(\Gamma) + \sum_{\Gamma \in \mathcal S_c}\frac{(2t)^{|\Gamma|}}{|\Gamma|!}w(\Gamma).
\end{align}
By Cor.~\ref{cor:decoupling corollary}, we can bound the first term in this equation by
\begin{align}
\sum_{\Gamma \in S_d}\frac{(2|t|)^{|\Gamma|}}{|\Gamma|!}w(\Gamma) \leq \prod_{i = 1,2} c |\partial B_i||\partial S_i| \mathrm{e}^{-\mu \mathsf d(S_i, B_i)}(e^{\mu \LRvel t}-1).
\end{align}
We are left with just the second term to bound.
For convenience, let 
\begin{align}
A(X \to Y;l) \equiv \{(U_1, \dots, U_l): U_1 \cap X \neq \emptyset, U_i \cap U_{i-1} \neq \emptyset, U_{l}\cap Y \neq \emptyset\}
\end{align}
where $X,Y,U_i \subseteq V$. Although not all $X,Y \in F$, we will bound the sum  by summing over all \textit{possible} subsets. Accordingly, given $\vec U = (U_1, \dots, U_l)$ where $U_i \subseteq V$, we define
\begin{align} 
\prod_{i}\Vert H_{U_i}\Vert\delta_{U_i \in F} \le \prod_i h e^{-\kappa |U_i|} \equiv W(\vec U)
\end{align}
Then we write as a shorthand for our bound on the single commutator
\begin{align}
\sum_{\substack{\Gamma(R \to S_i) \\ |\Gamma| = l}} w(\Gamma) 
 \leq \sum_{\vec U \in A(R \to S_i;l)}W(\vec U) \leq |\partial B_i| |\partial S_i|\frac{(h'K)^l}{K}\frac{\mathrm{e}^{-\mu\mathsf d(R, S_i)}}{\mathsf d(R, S_i)^\alpha}\equiv \frac{(h'K)^l}{K}Q_{i}
\end{align}
Note that in the first inequality, we have relaxed the bound by summing over non-irreducible (e.g. backtracking) paths $\vec U$ as well. By Proposition~\ref{lem:pathbreaking}, we can break each path up into the first factor $X$ which couples the two balls, and sequences of factors $\vec U$, $\vec V$ which then couple $X$ to $S_1$ and $S_2$ respectively. Then we account for all the ways that the elements of $\vec U$, $\vec V$ can be permuted without changing the ordering, which contributes a factor of $(l_1 + l_2)!/l_1!l_2!$:
\begin{align}
\label{eq:weightsumbound}
\sum_{\substack{\Gamma \in \mathcal S_c \\ |\Gamma| = l}} w(\Gamma) &\leq 
\sum_{l_1 + l_2 = l-1}
\sum_{P: \substack{P \cap B_1 \neq \emptyset\\ P \cap B_2 \neq \emptyset }}
W(P)\sum_{\substack{\vec U \in A(P \to S_1; l_1) \\ \vec V \in A(P \to S_2; l_2)}}\sum_{\text{ordering} \{U_i\}\cup \{V_i\}}W(\vec U \oplus \vec V)  \notag\\
&= \sum_{l_1 + l_2 = l-1}\frac{(l-1)!}{l_1!l_2!}\sum_{P: \substack{P \cap B_1 \neq \emptyset\\ P \cap B_2 \neq \emptyset }}
W(P)\sum_{\vec U \in A(P \to S_1; l_1)}W(\vec U)\sum_{\vec V \in A(P \to S_2; l_2)}W(\vec V).
\end{align}
Here we used $\oplus$ to emphasize that $\vec U$ and $\vec V$ are concatenated together as sequences of sets. Now we recognize that if $P$ intersects both $B_1$ and $B_2$, then we can pick $u \in \partial B_1 \cap P$ and $v \in \partial B_2 \cap P$ and
write $P = X\backslash\{u,v\} \sqcup Y \sqcup Z$, where $Y = P \cap B_1$, $Z = P \cap B_2$, and $X = P \cap R \cup \{u,v\}$. Then we can bound the sum over such $P$ by summing over $X,Y,Z$ satisfying $u,v \in X$, $u \in Y$, $v \in Z$. This is illustrated in Fig.~\ref{fig:longrangecoupling}.

\begin{figure}[t]
\centering
\includegraphics[width = .5\textwidth]{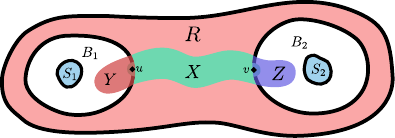}
\caption{Example of a long-range coupling term supported on $P = X \cup Y \cup Z$. Since $P$ must intersect the boundary of both $B_1$ and $B_2$ by the connectedness requirement, we pick two ``anchor" points $u,v$ on the respective boundaries of these regions so that $u \in X,Y$ and $v \in X,Z$.}
\label{fig:longrangecoupling}
\end{figure}

We have
\begin{align}
&\sum_{l_1 + l_2 = l-1}\frac{(l-1)!}{l_1!l_2!}\sum_{P: \substack{P \cap B_1 \neq \emptyset\\ P \cap B_2 \neq \emptyset }}
W(P)\sum_{\vec U \in A(P \to S_1; l_1)}W(\vec U)\sum_{\vec V \in A(P \to S_2; l_2)}W(\vec V) \notag\\
&\leq \sum_{l_1 + l_2 = l-1}\frac{(l-1)!}{l_1!l_2!}\sum_{\substack{u \in \partial B_1 \\ v \in \partial B_2}}\sum_{X \ni u, v}\sum_{Y \cap X =  u}\sum_{Z \cap X =  v} 
W(X \cup Y \cup Z) \sum_{\vec U \in A(Y \to S_1; l_1)}W(\vec U)\sum_{\vec V \in A(Z \to S_2; l_2)}W(\vec V) \notag\\
&\leq \frac{\mathrm{e}^{2\kappa}}{h^2}\sum_{l_1 + l_2 = l-1}\frac{(l-1)!}{l_1!l_2!}\sum_{\substack{u \in \partial B_1 \\ v \in \partial B_2}} \sum_{X \ni u,v}
W(X) \sum_{\vec U \in A(\{u\} \to S_1; l_1+1)}W(\vec U)\sum_{\vec V \in A(\{v\} \to S_2; l_2+1)}W(\vec V) 
\end{align}
We have abused notation slightly in writing \begin{equation}
    W(X \cup Y \cup Z) = W(X\backslash \{u,v\} \sqcup Y\sqcup Z ) = \frac{\e^{2\kappa}}{h^2}W(X)W(Y)W(Z).
\end{equation} 
Applying Lemma~\ref{lem:factorsum}, we bound $\sum_{X \ni u,v}w(X) \leq h'\mathrm{e}^{-\mu\mathsf d(u,v)} \leq h' \mathrm{e}^{-2\mu \chi}$. The remaining sums over paths are exactly the sums appearing in Lemma~\ref{lem:sumoverpaths}:
\begin{align}
&\frac{\mathrm{e}^{2\kappa}}{h^2}\sum_{l_1 + l_2 = l-1}\frac{(l-1)!}{l_1!l_2!}\sum_{\substack{u \in \partial B_1 \\ v \in \partial B_2}} \sum_{X \ni u,v}
w(X) \sum_{\vec U \in A(\{u\} \to S_1; l_1+1)}w(\vec U)\sum_{\vec V \in A(\{v\} \to S_2; l_2+1)}w(\vec V)  \notag\\
&\leq h'\frac{1}{h^2}\mathrm{e}^{2\kappa}\mathrm{e}^{-2\mu\chi}\sum_{l_1 + l_2 = l-1}\frac{(l-1)!}{l_1!l_2!}\frac{(h'K)^{l_1+1}}{K}Q_{1}\frac{(h'K)^{l_2 + 1}}{K}Q_2 \notag\\
&=  h'\qty(\frac{h'}{h})^2\mathrm{e}^{2\kappa}\mathrm{e}^{-2\mu\chi}Q_{1}Q_2\sum_{l_1 + l_2 = l-1}\frac{(l-1)!}{l_1!l_2!}(h'K)^{l_1+l_2}
\end{align}
Now we can plug this bound into the left-hand side of Eq.~\eqref{eq:coupled_commutator_bound}:
\begin{align}
\sum_{\Gamma \in \overline{\mathcal S}}\frac{(2t)^{|\Gamma|}}{|\Gamma|!}W(\Gamma) \leq 
 h'\sum_{l=1}^{\infty}\sum_{l_1 + l_2 = l-1}\frac{(2t)^{l}}{l!}\frac{(l-1)!}{l_1!l_2!}(h'K)^{l_1}(h'K)^{l_2}\prod_{i  = 1,2} b \mathrm{e}^{-\mu \chi+\kappa}Q_i.
\end{align}
We can similarly rewrite this sum in terms of sums over $l_1$, $l_2$ alone, i.e.
\begin{align}
\sum_{l=1}^{\infty}\sum_{l_1 + l_2 = l-1}\frac{(2t)^{l}}{l!}\frac{(l-1)!}{l_1!l_2!}(h'K)^{l_1}(h'K)^{l_2} &= 2t\sum_{l_1 = 0}^\infty \sum_{l_2 = 0}^\infty\frac{(2h'Kt)^{l_1}}{l_1!}\frac{(2h'Kt)^{l_2}}{l_2!}\frac{1}{(l_1+l_2+1)} \notag\\
& \leq 2t\qty(\sum_{l_1=0}^{\infty}\frac{(2h'Kt)^{l_1}}{l_1!})\qty(\sum_{l_2 = 0}^\infty\frac{(2h'Kt)^{l_2}}{l_2!})\notag\\
&= 2t\mathrm{e}^{2\mu \LRvel t},
\end{align}
where we used $v_{\mathrm{LR}}=\frac{2Kh'}{\mu}$. Now with $K > 1$, our bound simplifies to
\begin{align}
\sum_{\Gamma \in \mathcal S_c}\frac{(2t)^{|\Gamma|}}{|\Gamma|!} w(\Gamma)\leq \mu \LRvel t\mathrm{e}^{-2\mu \chi + 2\kappa}\prod_{i  = 1,2} b|\partial B_i||\partial S_i| \mathrm{e}^{\mu (\LRvel t - \mathsf d(S_i, B_i) )}
\label{eq:coupled_commutator_bound}
\end{align}
This gives a bound on the second term in Eq.~\eqref{eq:breaksum}. To obtain a simple expression for the sum,
we observe that $\mathrm{e}^{t}-1 \leq t\mathrm{e}^{t}$, which follows from a comparison of their Taylor series.
Then we have $b > 1$ because $h' > h$, and $c = \frac{1}{K} < 1$ by construction, so
\begin{align}
\prod_{i = 1,2} c|\partial B_i||\partial S_i|\mathrm{e}^{-\mu\mathsf d(S_i, B_i)}\left(\mathrm{e}^{\mu \LRvel t}-1\right) \leq
\prod_{i = 1,2}b |\partial B_i||\partial S_i|\mathrm{e}^{-\mu\mathsf d(S_i, B_i)}(\mu \LRvel t)e^{\mu \LRvel t}.
\end{align}
Adding this to Eq.~\eqref{eq:coupled_commutator_bound} gives the desired bound.
\end{proof}
We remark that, since $\kappa \sim \mu$, if $\mu\chi > \kappa$, then as $\kappa \to \infty$, the factor $e^{2(\kappa-\mu\chi)}$ suppresses the contribution of the terms in $\mathcal S_c$, and we recover the product bound for the strictly local Hamiltonian.

\subsection{Nested commutators: the general case}
Now the goal is to generalize this proof to arbitrary $m$. Most of the manipulations are similar to those from the previous section. The problem now is that we need to account for long-range coupling between all possible subsets of the regions.
If we were to naively sum the long-range couplings, the number of such subsets would grow much faster than the exponential suppression. In order to overcome this difficulty, we need to capture the fact that there are only very few small connected subsets coupling the regions, and the contribution of the large connected subsets are suppressed exponentially in the distance between these subsets. In order to formalize this notion, we note that in $d$ spatial dimensions, we can tile the plane with Euclidean boxes such that every box is adjacent to $3^d - 1$ boxes. We specialize to the case where each box contains at most one region $B_1, \dots, B_m$, and the distance between $B_i$ and any other box is at least $\chi$. This is straightforwardly generalized to any way $\mathbb R^d$ can be broken into a course-grained lattice by tiling it with arbitrary shapes. 
\begin{defn}
\label{def:coupling}
If $[T] \in \mathcal S_c$ is an equivalence class of trees with irreducible paths $\Gamma_1, \dots, \Gamma_m$ such that $\Gamma_{1} \cap \Gamma_{2} \cap \dots \cap \Gamma_{k} \neq \emptyset$ and $\Gamma_{i} \cap \Gamma_{j} = \emptyset$ for $i \leq k$ and $j > k$, then we say that $[T]$ couples $B_{1}, \dots, B_{k}$. 
\end{defn}
\begin{prop}
Let $T$ be a causal tree coupling $B_1, \dots, B_k$. Let $\Gamma_1, \dots, \Gamma_k$ be the irreducible paths coupling $S_1, \dots, S_k$ to $R$. Then $\Gamma_1 \cap \dots \cap \Gamma_k = \{X\}$, where $X$ intersects with $B_1, \dots, B_k$, and $X$ is the first element of $\Gamma_1, \dots, \Gamma_k$. Furthermore, all elements of $\Gamma_i\backslash \{X\}$ are contained within $B_i$ for each $i\le k$.
\end{prop}
\begin{proof}
By the algorithm used to construct causal trees, the first element in each irreducible path is the only one to intersect $R$. Since $X$ must be connected and couples the regions, it must intersect with $R$. Therefore, removing $X$ from $\Gamma_i$ creates a path that lies entirely within $B_i$.
\end{proof}
\begin{defn}
We call a connected set of boxes a \textit{connected cluster}. For any $k \leq m$, a connected cluster containing $B_1, \dots, B_k$ which has minimal volume is called a \textit{minimal cluster} corresponding to $B_1, \dots, B_k$. An example is illustrated in Fig.~\ref{fig:fig_commutator}.
\end{defn}
We then proceed by induction on $m$. If a number of regions are coupled together as shown in Fig.~\ref{fig:fig_commutator} by a long-range term in a particular equivalence class, then, by applying Cor.~\ref{cor:decoupling corollary}, we can factorize the system into two disjoint subsets of smaller size, which will establish the recursion.
\begin{figure}[t]
    \centering
    \includegraphics[width=0.8\linewidth]{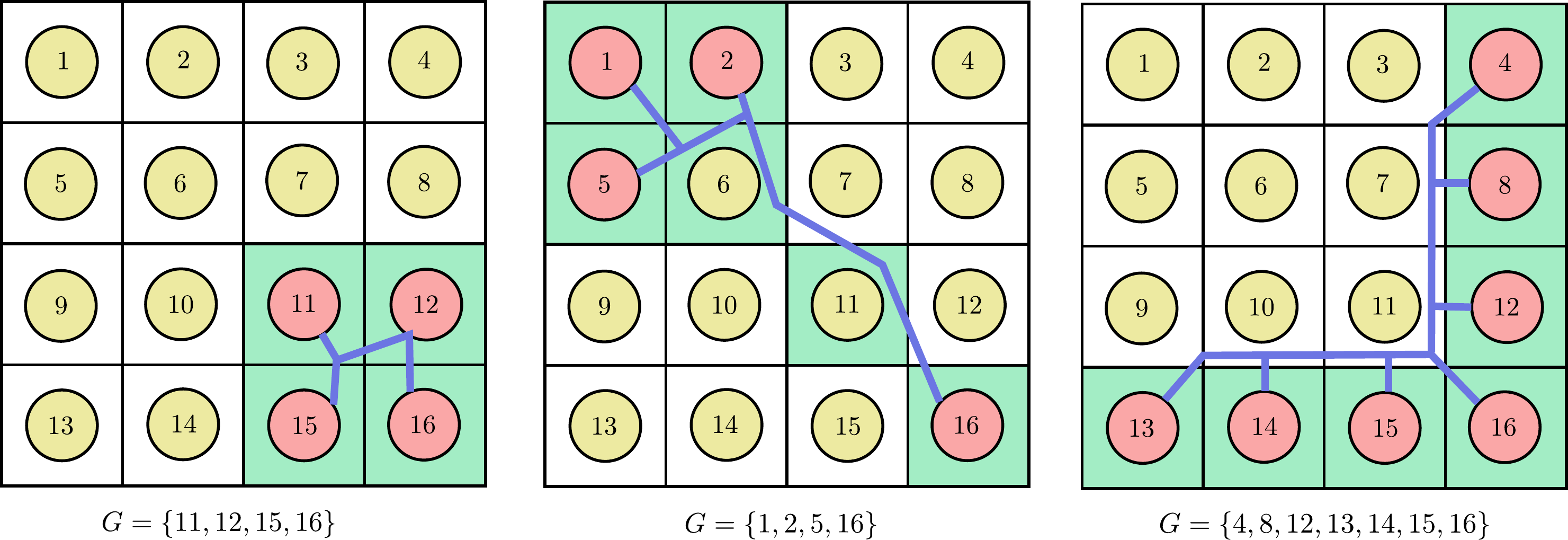}
    \caption{Here $\mathbb R^2$ is covered with Euclidean boxes, and at most one region of interest $B_i$ is situated within each box with a distance at most $\chi$ from the edge of the box. Subset $G$ is a subset of the regions representing the ways region $16$ can be coupled by a single term in a given irreducible path to the rest of the system, as shown in red. This coupling term is illustrated in blue. Since the coupling term is connected, it must pass through boxes that form a connected cluster; $G$ itself does not need to be connected. The green boxes show a minimal connected cluster containing $G$, which captures the exponential suppression of the terms coupling together the regions. We denote the number of green boxes by $M_{\mathrm{min}}(G)$. Our formalism accounts for terms like the one shown on the left that couple nearby balls as well as ones like the two shown on the right, which couple regions far away.}
    \label{fig:fig_commutator}
\end{figure}
We first derive the contribution from an equivalence class which couples $k$ regions using techniques from the previous section: 
\begin{lem}
\label{lem:coupled_class_exp_bound}
Suppose that $\mathcal S_c$ is the set of equivalence classes coupling $B_1, B_2, \dots, B_k$ with minimal cluster of order $M_{\mathrm{min}}$ and irreducible skeleton of length $l$. 
Then we can carry out the weighted sum over $\mathcal S_c$ with the bound
\begin{align}
\sum_{l=1}^{\infty}\sum_{\substack{\Lambda \in \mathcal S_c \\ |\Lambda| = l}}\frac{(2t)^l}{l!}w(\Lambda) \leq (\mu \LRvel t)\mathrm{e}^{-\mu \chi M_{\mathrm{min}}}  \prod_i \mathrm{e}^{\kappa}b|\partial B_i||\partial S_i|\mathrm{e}^{\mu(\LRvel t- \mathsf d(S_i, B_i))}.
\end{align}

\end{lem}
\begin{proof}
Fix an equivalence class $\mathcal S_c$ that couples $B_1, \dots, B_k$. As in the proof of Prop.~\ref{prop:double_commutator_bound}, we can bound the sum by
\begin{align}
\sum_{\substack{\Lambda \in \mathcal S_c \\ |\Lambda| = l}}w(\Lambda) &\leq \sum_{\sum_k l_k = l-1}\frac{(l-1)!}{l_1! \cdots l_k!}\sum_{X \cap B_1 \dots X\cap B_k \neq \emptyset}\Vert H_X\Vert\sum_{\vec U_1 \in A(X \to S_1;l_1)}W(\vec U_1) \dots \sum_{\vec U_k \in A(X \to S_k; l_k)} W(\vec U_k) \notag\\
&\leq \sum_{\sum_k l_k = l-1}\frac{(l-1)!}{l_1! \cdots l_k!}\sum_{u_1 \in \partial B_1 \dots u_k \in \partial B_k}\sum_{X \ni u_1, \dots, u_k}W(X \backslash \{u_1, \dots, u_k\})\notag\\
&\hspace{2cm}\times\sum_{\vec U_1 \in A(\{u_1\} \to S_1;l_1)}W(\vec U_1) \cdots \sum_{\vec U_k \in A(\{u_k\} \to S_k; l_k)} W(\vec U_k)\notag \\
&\leq h' \qty(\frac{h'}{h})^k \mathrm{e}^{k\kappa}\mathrm{e}^{-\mu M_{\mathrm{min}}\chi}Q_1 \cdots Q_k\sum_{\sum_i l_i = l-1}\frac{(l-1)!}{l_1! \cdots l_k!}(h'K)^{l-1}
\end{align}
We have skipped a few steps because the manipulations are the same as in Proposition~\ref{prop:double_commutator_bound}. The significant difference is that each path needs to pass through a connected cluster, and the minimal cluster captures the number of boxes that it must either pass through or connect to the ball at the center. Therefore if the minimal cluster is of size $M_{\text{min}}$, the weight of this path is bounded by $W(X \backslash \{u_1, \dots, u_k\}) \leq h\mathrm{e}^{k\kappa}\mathrm{e}^{-\mu \kappa M_{\text{min}}}$, from which the bound on the summation of all possible subsets containing $u_1 \in \partial B_1, \dots, u_k \in \partial B_k$ follows as in Lemma~\ref{lem:factorsum}.
With this, we have
\begin{align}
\sum_{l=1}^{\infty}\sum_{\substack{\Lambda \in \mathcal S \\ |\Lambda| = l}}\frac{(2t)^l}{l!}w(\Lambda) 
&\leq 
h'\mathrm{e}^{-\mu \chi M_{\mathrm{min}}}\mathrm{e}^{k\kappa}\qty(\prod_i b|\partial B_i||\partial S_i|e^{-\mu\mathsf d(S_i, B_i)})\sum_{l=1}^{\infty} \frac{(2t)^{l}}{l!}\qty(\sum_{l_1 + \dots + l_k = l-1}\frac{(l-1)!}{l_1!l_2! \cdots l_l!}\prod_i(h'K)^{l_i}).
\label{eq:timebound}
\end{align}
Then we can bound the compound sum with a product:
\begin{align}
\sum_{l=1}^\infty\frac{(2t)^l}{l!}\sum_{l_1 + \dots + l_k = l-1}\frac{(l-1)!}{l_1!l_2! \cdots l_l!}\prod_i(h'K)^{l_i} &=
2th'\sum_{l_1 + \dots + l_k = l-1} \frac{(2t)^{\sum_{i}l_i}}{l_1! \cdots l_k!}\frac{1}{1+\sum_{i}l_i}\prod_i(h'K)^{l_i} \notag \\
&\leq \mu \LRvel t \qty(\sum_{l = 0}^{\infty}\frac{(2h'Kt)^{l}}{l!})^k \notag \\
&\leq \mu \LRvel t \mathrm{e}^{k\mu \LRvel t}.
\end{align}
Inserting this into \eqref{eq:timebound} gives the desired result.
\end{proof}

\begin{thm}\label{thm:exponential_H}
Consider a quasi-local interaction defined on a graph in $d$ dimensions. Suppose that we tile the plane with Euclidean boxes, where each box has $3^d-1$ adjacent neighbors. Assume that $B_1, \ldots ,B_m$ are regions, each contained within a separate box such that the distance between $B_i$ and any other box is at least $\chi$, where $\mu \chi > \max(\log(2)+d\log(3)+2, \kappa)$, and $S_i\subset B_i$. Then the nested commutator bound
\begin{align}\label{eq:exponential_H_C<}
\frac{1}{2^m}\Vert [\mathcal O_m, [\dots, [\mathcal O_1, A(t)]\dots ]]\Vert \leq 
\mu v t (\mathrm{e}^{-\gamma + \kappa}+\mu v t)^{m-1}\prod_{i=1}^m |\partial B_i||\partial S_i|b \mathrm{e}^{\mu(\LRvel t- \mathsf d(S_i, R))}
\end{align}
holds,
where $A$ is supported in $R$ which is the complimentary region to $\bigcup_i B_i$, $\mathcal O_i$ is supported within $S_i$, $\gamma = \mu \chi - \log(2\e3^d) > 1$, and $v = 4\e3^d\LRvel$.
\end{thm}

\begin{proof}
Let $B_1, \dots, B_m$ be defined as in the theorem statement. We then induct on $m$. The base case was proved in Sec.~\ref{sec:single_commutator_exp_bound}. We then consider an additional region $B_{m+1}$.
We sort the equivalence classes by considering all the possible ways $B_{m+1}$ can be coupled to $B_{1}, \dots, B_{m}$. For any $k \leq m$, let $\Lambda(\{i_1, \dots, i_k\})$ denote the set of irreducible skeletons connecting $S_{i_1}, \dots, S_{i_k}$ to $R$ and let $\Gamma(\{i_1, \dots, i_k\})$ denote the irreducible skeletons coupling $i_1, \dots, i_k$. As we have shown in Eq.~\eqref{eq:productbound}, we can bound the nested commutator with 
\begin{align}
\frac{1}{2^m}\Vert[\mathcal O_{m+1},[\dots,[\mathcal O_1, A(t)]]\dots]\Vert &\leq \sum_{\Lambda \in \Lambda(1, \dots, m)}\frac{(2|t|)^{|\Lambda|}}{|\Lambda|!} w(\Lambda) \notag\\
&\leq \sum_{\substack{G \subseteq \{1, \dots, m, m+1\} \\ m+1 \in G}}\qty(\sum_{\Gamma \in \Gamma(G)}\frac{(2|t|)^{|\Gamma|}}{|\Gamma|!}w(\Gamma))\qty(\sum_{\Lambda \in \Lambda(G^c)}\frac{(2|t|)^{|\Lambda|}}{|\Lambda|!}w(\Lambda)).
\end{align}
The decomposition above is illustrated in Fig.~\ref{fig:fig_commutator}, with $G$ represented by the red circles (all connected to $B_{m+1}$) and $G^c$ being the yellow circles (not connected to $B_{m+1}$, not necessarily all connected to each other). Let $C(G) \equiv \sum_{\Lambda \in \Lambda(G)}\frac{(2|t|)^{|\Lambda|}}{|\Lambda|!}w(\Lambda)$ with $C(\emptyset) = 1$. Applying Lemma~\ref{lem:coupled_class_exp_bound}, we have a bound for the first term:
\begin{align}
C(\{1, \dots, m, m+1\}) \leq \sum_{\substack{G \subseteq \{1, \dots, m, m+1\} \\ m+1 \in G}}\mu \LRvel t \mathrm{e}^{-\mu \chi(M_{\mathrm{min}}(G)-1)}C(G^c)\prod_{i \in G}|\partial B_i||\partial S_i|b\mathrm{e}^{\kappa (|G|-1)}\exp(\mu(\LRvel t - r_i))
\end{align}
with $r_i = \mathsf d(S_i, R)$ for short. Since we assumed that $\mu \chi > \kappa$, and $|G| \leq M_{\mathrm{min}}(G)$ we relaxed $\mathrm{e}^{\kappa|G|}\mathrm{e}^{-\mu \chi M_{\mathrm{min}}(G)}$ to $\mathrm{e}^{\kappa(|G|-1)}\mathrm{e}^{-\mu \chi (M_{\mathrm{min}}(G) - 1)}$. This choice greatly simplifies the induction because it makes Lemma~\ref{lem:coupled_class_exp_bound} consistent with the $m=1$ case that we proved in Sec.~\ref{sec:single_commutator_exp_bound}, thereby serving as a base case.
Then we introduce $D(G)$ such that 
\begin{equation}
C(G) = \prod_{i\in G} b |\partial B_i||\partial S_i|\exp(\mu(\LRvel t - r_i))D(G).
\end{equation}
Then $D(G)$ satisfies the relationship
\begin{align}\label{ddef}
D(\{1, \dots, m, m+1\}) \leq (\mu \LRvel t)\sum_{\substack{G \subseteq \{1, \dots, m,m+1\} \\ m+1 \in G}}\mathrm{e}^{(|G|-1)\kappa}\mathrm{e}^{-\mu \chi (M_{\mathrm{min}}(G)-1)}D(G^c).
\end{align}
Now let \begin{equation}
    Q(k) = \sup_{\substack{G \subseteq \{1, \dots, m, m+1\} \\ |G| = k}}D(G).
\end{equation} 
Note that $Q(m+1)=D(\{1,\dots,m,m+1\})$ and $Q(0) = 1$. This gives us the simpler but weaker bound
\begin{align}\label{qdef}
 Q(m+1) \leq (\mu \LRvel t)\sum_{\substack{G \subseteq \{1, \dots, m, m+1\} \\ m+1 \in G}}\mathrm{e}^{(|G|-1)\kappa}\mathrm{e}^{-\mu \chi (M_{\mathrm{min}}(G)-1)}Q(m+1-|G|).
\end{align}
We can then bound the sum above by summing instead over connected clusters of boxes and then all possible subsets of those clusters:
\begin{align}
\sum_{\substack{G \subseteq \{1, \dots, m, m+1\} \\ G \ni m+1}}&\mathrm{e}^{(|G|-1)\kappa}\mathrm{e}^{-\mu \chi (M_{\mathrm{min}}(G)-1)}Q(m+1-|G|) \notag\\
&= \mathrm{e}^{\mu \chi}\sum_{k=1}^{m+1}\mathrm{e}^{(k-1)\kappa}\sum_{M = k}^\infty\sum_{\substack{\text{connected $X \ni m+1$} \\ |X| = M}}\sum_{\substack{G \subseteq X \\|G| = k}}\mathrm{e}^{-\mu \chi M}Q(m+1-k) \notag \\
&= \mathrm{e}^{\mu \chi}\sum_{k=1}^{m+1}\mathrm{e}^{(k-1)\kappa}\sum_{M = k}^\infty(\e 3^d)^M\binom{M}{k}\mathrm{e}^{-\mu \chi M}Q(m+1-k) \notag \\
&\leq \mathrm{e}^{\mu \chi}\sum_{k=1}^{m+1}\mathrm{e}^{(k-1)\kappa}\sum_{M = k}^\infty \mathrm{e}^{-M\gamma}Q(m+1-k) \notag \\
&\leq 4\e 3^d\sum_{k=0}^{m}\mathrm{e}^{-k(\gamma-\kappa)}Q(m-k).
\end{align}
Here we set $\gamma = \mu \chi - \log(2\e 3^d)$ and require that $\gamma > 1$, which we can always ensure by choosing $\chi$ large enough. To obtain the third line, we used Prop.~\ref{prop:clustercounting} to bound the number of clusters, and in the fourth line we bounded $\binom{M}{k} \leq 2^M$. With $Q(0)= 1$ and $Q(1) \leq \mu \LRvel t \leq \mu v t$ already proven, for our inductive hypothesis we assume that, for $0 < k \leq m$,
\begin{align}
Q(k) \leq \mu vt (\mu vt + \mathrm{e}^{-(\gamma - \kappa)})^{k-1}.
\end{align}

For convenience, let $t' \equiv \mu v t$ and $a \equiv \mathrm{e}^{-(\gamma - \kappa)}$. Then
\begin{align}
Q(m+1) &\leq t'a^m+t' \sum_{k=1}^{m}a^{m-k}t'(t' + a)^{k-1} = t'a^m + \sum_{k=1}^{m}a^{m-k}\sum_{j=0}^{k-1}\binom{k-1}{j}t'^{j+2}a^{k-1-j} \notag \\
&= t'a^m+\sum_{j=0}^{m-1}a^{m-j-1}t'^{j+2}\sum_{k=j+1}^{m}\binom{k-1}{j} =t'a^m+\sum_{j=0}^{m-1}a^{m-j-1}t'^{j+2}\binom{m}{m-j-1} \notag \\
&= t'\sum_{j=0}^{m}a^{m-j}t'^{j}\binom{m}{m-j} = t'(t' + a)^m.
\end{align}
This completes the induction. We used the ``hockey stick" relation $\sum_{m=k}^n \binom{m}{k} = \binom{n+1}{k+1}$ for the binomial coefficients. Unraveling the definition of $Q$,
\begin{align}
\frac{1}{2^m}\Vert [\mathcal O_m,[\dots, [\mathcal O_1, A(t)]]\dots]\Vert &\leq 
Q(m)\prod_{i}b|\partial B_i||\partial S_i| \exp(\mu(\LRvel t - r_i)) \notag\\ &\leq 
\mu v t (\mu v t + \mathrm{e}^{-\gamma + \kappa})^{m-1}\prod_{i}b |\partial B_i||\partial S_i| \exp(\mu(\LRvel t - r_i))
\end{align}
This is the desired result.
\end{proof}
We remark on some of the properties of the above formula. Firstly, we see that, unlike the strictly local case, this bound is $\mathrm{\Theta}(t)$ for small $t$. 
This reflects the fact that there are now terms at first order that couple the $m$ regions to $R$. The term in the sum proportional to $t$ has a prefactor of $\mathrm{e}^{-(m-1)(\gamma-\kappa)}$. The $m-1$ instead of $m$ comes from the fact that we needed to relax the suppression in $\chi$ slightly to obtain this simple formula, but this reflects the fact that any term which directly couples all $m$ regions is exponentially suppressed in $m$. For example, expanding $Q(3) = t'^3 + 2t'^2 a + a^2t'$, we can identify the first term as decoupled, the second as originating from a coupling of two regions with the binomial coefficient counting the number of such couplings, and the last term representing a coupling between all three. Next, when $\mu\chi \gg \kappa, \log(2e3^d)$, the prefactor approaches $(\mu v t)^m$. This can be seen as the decoupled limit, and reflects the fact that one has to go to $m^{\text{th}}$ order to connect each $S_i$ to $R$, because irreducible skeletons $\Lambda$ with $|\Lambda| < m$ are exponentially suppressed. This also happens when we take $\kappa \to 0$, and, since $\mu \sim \kappa$, the requirement that $\mu \chi > \kappa$ is well behaved in this limit for fixed $\chi$. Lastly, we note that this bound is uniform in the number of sites in the system $|V|$, which permits us to pass to the thermodynamic limit. 

Although the above bound is stronger, we can also obtain a simplified result that more closely resembles standard Lieb-Robinson bounds.
\begin{cor}[Thm.~\ref{thm:exponential_H} simplified]
\label{thm:exponential_H_simplified}
Under the same assumptions as Thm.~\ref{thm:exponential_H}, the following bound
\begin{align}
\frac{1}{2^m}\Vert [\mathcal O_m, [\dots, [\mathcal O_1, A(t)]\dots ]]\Vert \leq 
\prod_{i=1}^m b|\partial B_i||\partial S_i| \mathrm{e}^{\mu(v t- \mathsf d(S_i, R))}
\end{align}
holds, 
where $A$ is supported in $R$ which is the complimentary region to $\bigcup_i B_i$, $\mathcal O_i$ is supported within $S_i$, $\gamma = \mu \chi - \log(2\e 3^d) > 1$, and $v = 8\e3^d\LRvel$.
\end{cor}

\begin{proof}
We note that $\e^{-\gamma + \kappa} < 1$, so we can apply the bound $(\e^{-\gamma + \kappa} + \mu v t)^{m-1} < \e^{(m-1)\mu v t}$ to Thm.~\ref{thm:exponential_H}. Using $\mu vt < \e^{\mu v t}$ and $v = 4\e 3^d \LRvel > \LRvel$ (with $v$ defined as in Thm.~\ref{thm:exponential_H}) gives the simplified bound.
\end{proof}

The conclusion of Cor.~\ref{cor:volumescaling} also holds for quasi-local interactions:
\begin{cor}
\label{cor:volumescaling_exp}
Suppose that $H$ is a quasi-local Hamiltonian. Let $B_R$ denote a metric ball of radius $R$. There exist constants $\gamma, c_{\rm{LR}} > 0$ such that we can choose $v_1, \dots, v_m \in B_R$ (where $m$ depends on $R$) for which
\begin{align}
C_{v_1, \dots, v_m}^{B_R^c}(t)  \leq c_{\rm{LR}}\exp(-\gamma\frac{(R-\LRvel t)^d}{(\LRvel t)^{d-1}})
\end{align}
for any $\LRvel t > 1$.
\end{cor}

\begin{proof}
First, consider the interaction graph embedded in $\mathbb R^d$. Let $\chi \equiv \max(\kappa, \log[2\e^23^d])/\mu$. There exists a constant $B$ such that, for $\xi > 1$, we can tile $\mathbb R^d$ with Euclidean boxes where each box contains a ball of radius $\xi + \chi$ and is contained within a ball of radius $B \xi$. Let $C_1$ be a constant such that $|\partial B_{\xi}| < (C_1 \xi)^{d-1}$ for all $\xi > 1$, and let $\xi \equiv \alpha v t$ for some $\alpha > 1$.  As in Cor.~\ref{cor:volumescaling}, there is a constant $C_2$ such that, for each $\alpha > 1$, we can find $K \gg 1$ such that for all $R > K \LRvel t$ we can find $m > C_2R^d/\xi^d$ of these boxes lying entirely within $B_R$. Then from Cor.~\ref{thm:exponential_H_simplified}, we have
\begin{align}
C_{\vec v}^R \leq \qty((C_1 \xi)^{d-1} \e^{\mu(2vt - \xi)})^{C_2\frac{R^d}{\xi^d}}
\end{align}
 The bound above can be re-written as
\begin{align}
 \qty((C_1 \xi)^{d-1} \e^{\mu(2vt - \xi)})^{C_2\frac{R^d}{\xi^d}} &=
\qty((C_1 \alpha vt)^{\frac{d-1}{\alpha vt}} \mathrm{e}^{\mu\frac{(2 - \alpha)}{\alpha}})^{C_2\frac{R^d}{(\alpha v t)^{d-1}}}
\end{align}
Taking the logarithm of the term in parentheses,
\begin{align}
\log((C_1 \alpha vt)^{(d-1)/\alpha vt} \mathrm{e}^{\mu(2 - \alpha)/\alpha}) \leq \frac{1}{\alpha}\qty[\frac{d-1}{\mathrm e} + (d-1)\log(C_1\alpha) + 2\mu]-\mu \equiv -\beta
\end{align}
Since the first term on the right-hand side vanishes as $\alpha \to \infty$, it is clear that we can choose $\alpha$ large enough such that $\beta > \frac{\mu}{2}$. 
Thus the bound becomes
\begin{align}
C_{\vec v}^R \leq \mathrm{e}^{-\frac{C_2\mu}{2\alpha^{d-1}}\frac{R^d}{(v t)^{d-1}}}
\end{align}
With this demonstrated for $R > K \LRvel t$, the rest of the proof is the same as in Cor.~\ref{cor:volumescaling}, save for the fact that we must choose a smaller $\mu$ to achieve the same area prefactor in Prop.~\ref{prop:quasilocalbound} as in Thm.~\ref{thm:standardLR}.
\end{proof}

\section{Accuracy of classical simulations of quantum dynamics\label{sec:simulation}}

In random unitary circuits \cite{nahum2018,keyserlingk2018}, the dynamics are restricted by a light cone with volume $\sim t^d$. The complexity of simulating these dynamics by truncating the operator outside of the light cone is exactly bounded by $\exp(t^d)$. This intuition breaks down for locally generated continuous-time evolution, where the light cone is not sharp.  As we discussed in the Introduction, one then needs to specify an error tolerance $\epsilon$ to which the simulation must be performed. Using classical Lieb-Robinson bounds, one expects that the truncation algorithm with tolerance $\epsilon$ will produce a quasi-polynomial error growth $\sim \exp(\log(\epsilon)^d)$, which originates from failing to control the magnitude of the operator tails outside the light cone. In Ref.~\cite{cluster_Alhambra23}, it was shown that, for product initial states and strictly local interactions, a simulation algorithm based on truncating the support of an operator to a large subset achieves the optimal error scaling $\epsilon^{-\exp(t)}$, which is polynomial in $\epsilon$, but seemingly suboptimal in $t$ in finite dimensions. 

The bound achieved by Wild et al. in \cite{cluster_Alhambra23} draws on previous cluster expansion techniques from statistical physics. Such approximation methods involve an asymptotic expansion in the local parameters of the Hamiltonian, which work as long as the Hamiltonian is local on a bounded-degree graph. 
By contrast, we use the machinery developed above to revisit classical simulation of quasi-local dynamics from the perspective of an expansion in \textit{time}, and show that our algorithm is able to optimally leverage spatial locality and the finite dimensionality of the lattice.

\subsection{Operator expansion with exponential-in-volume tails}
In order to bound the simulability of quantum systems with quasi-local interactions (which also includes strictly local interactions), we wish to have better control over the fraction of an operator which is supported on \emph{any} large set $S$ containing $\gg (v_{\mathrm{LR}}t)^d$ vertices.  We have already made progress along these lines.  In Corollary~\ref{cor:volumescaling} and Corollary~\ref{cor:volumescaling_exp}, we made a weaker statement about the tail of an operator supported on vertices within a \emph{specific} large set $S$.  What remains is to show that, even though there are $\sim \exp(|S|)$ possible choices of connected sets with the volume of $S$---a factor which is large enough to easily overwhelm the volume-law Lieb-Robinson bound in either corollary---the \emph{net} contribution of all such large operators is still exponentially small in volume.  This is achieved by the following theorem, where, for technical convenience, we restrict to dynamics on a square lattice.

\begin{thm}\label{thm:expansion}
Let $H$ be a quasi-local Hamiltonian on a $d$-dimensional square lattice $V$ of qudits, whose Lieb-Robinson bound is given by Proposition~\ref{prop:quasilocalbound} with parameters $\mu,v_{\rm LR}$. Consider a local operator $A$ with $\norm{A}=1$ acting on a connected region of size $\mathrm{O}(1)$, which is evolved to $A(t)$ at time $t$ by $H$. It can be expanded to connected subsets $S$ that contain the support of $A$: \begin{equation}
    A(t) = \sum_{S\subset V} \widetilde{A}(S;t).
\end{equation}
Here $\widetilde{A}(S;t)$ is supported in $S$,
such that all contributions from $|S|>M$ decay exponentially with the volume cutoff $M$: for any $M\ge 1$,
\begin{equation}
    \widetilde{A}(t) = \sum_{\substack{S\subset V:|S|\le M\\ \supp(A) \subseteq S}} \widetilde{A}(S;t),
\end{equation}
approximates $A(t)$ with error \begin{equation}\label{eq:A-tA_thm}
    \norm{A(t)-\widetilde{A}(t)} \le c_d \exp\qty[4\mu v_{\rm LR}t-c_d' \mu \frac{M}{(v_{\rm LR}t+c_{\rm box})^{d-1}}],
\end{equation}
where $c_d,c_d',c_{\rm box}$ are constants determined by $d,\mu,v_{\rm LR}$.
\end{thm}

\begin{proof}
Similar to the previous section, we partition the lattice sites into non-overlapping boxes $V=\boxx_0\sqcup \boxx_1\sqcup\cdots,$ where each box $\boxx$ is a cube of length $r$ (thus containing $r^d$ sites), and operator $A$ is contained in one box $\boxx_0$. We choose $r$ to be \begin{equation}\label{eq:rbox=t}
    r=4(v_{\rm LR}t+c_{\rm box}) > \max\qty[d,\; \frac{45}{4\mu}3^d (\log 2+d\log 3)],
\end{equation}
where the inequality holds by choosing a sufficiently large $\mathrm{O}(1)$ constant $c_{\rm box}$. The boxes can be viewed as vertices of a coarse-grained $d$-dimensional square lattice, and form a graph $G_{\rm b}=(V_{\rm b}, E_{\rm b})$ where two boxes $\boxx,\boxx'\in V_{\rm b}$ share an edge in $E_{\rm b}$ if and only if they are neighbors, because $\mathsf{d}(\boxx,\boxx')\le d$.  Here $\mathsf{d}(\cdot,\cdot)$ is the distance function on the original lattice $V$. Note that here we do not work with factor graphs directly.
According to the Taylor expansion \eqref{eq:resum}, we organize the time-evolved operator $A(t)=\mathrm{e}^{t\mathcal{L}}A$ by the boxes it touches [or has touched in the past by taking commutators with factors $X_j$ in the causal tree $(X_1,\cdots,X_n)$]:
\begin{equation}\label{eq:At=clusters}
    A(t)=\sum_{\text{connected }S_{\rm b}\subset V_{\rm b}: S_{\rm b}\ni \boxx_0} A(S_{\rm b};t),
\end{equation}
with \begin{equation}
    |A(S_{\rm b};t))\equiv \sum_{n=0}^\infty  \frac{t^n}{n!}\sum_{X_1, \dots, X_n \in F:\;\qty( \boxx \in S_{\rm b} \; \Leftrightarrow \; \exists X_j \text{ such that } X_j\cap \boxx\neq \emptyset )}\mathcal L_{X_n}\cdots \mathcal L_{X_1}|A).
\end{equation}
The summation here is taken over sets of adjacent boxes $S_{\rm b}\subset V_{\rm b}$ that are connected in $G_{\rm b}$, which we call a \emph{connected cluster}.\footnote{This definition may differ slightly from the one used in the literature, where a cluster typically means a multiset of subsets of the graph. We use the term here to illustrate the connection between our algorithm and other cluster expansion methods.}
In other words, we include the causal tree $(X_1,\cdots,X_n)$ in $A(S_{\rm b};t)$ if and only if the factors touch all boxes in $S_{\rm b}$ but no boxes in its complement $S_{\rm b}^{\rm c}$.
Observe that we have restricted to subsets of boxes $S_{\rm b}$ that are connected in $G_{\rm b}$ and that contain $\boxx_0$; otherwise, $A(S_{\rm b};t)=0$ simply because the operator grows in a connected way.  

An equivalent way to express $A(S_{\rm b};t)$ is by induction: the starting point is \begin{equation}\label{eq:Ab0=}
    A(\{\boxx_0\};t)=e^{t\mathcal{L}_{\boxx_0}} A,
\end{equation}
where $\mathcal{L}_S:=\sum_{X\in F:X\subset S}\mathcal{L}_X$.
Then, given all $A(S_{\rm b};t)$ for $|S_{\rm b}|\le m$, any $S'_{\rm b}$ with size $|S'_{\rm b}|= m+1$ is given by \begin{equation}\label{eq:AS'b=ASb}
    A(S'_{\rm b};t) = e^{t\mathcal{L}_{S'_{\rm b}}} A - \sum_{S_{\rm b}\subsetneqq S'_{\rm b}:\, S_{\rm b}\neq \emptyset}A(S_{\rm b};t),
\end{equation}
where $\mathcal{L}_{S'_{\rm b}}:=\mathcal{L}_{\boxx_0\cup \boxx_1\cup\cdots\cup \boxx_m}$ for $S'_{\rm b}=\{\boxx_0,\boxx_1,\cdots,\boxx_m\}$. \eqref{eq:AS'b=ASb} holds because when Taylor-expanding $e^{t\mathcal{L}_{S'_{\rm b}}} A$, any causal tree $(X_1,\cdots,X_n)$ that does not touch all boxes in $S'_{\rm b}$ is canceled by a corresponding term in exactly one $A(S_{\rm b};t)$ in the second term of Eq.~\eqref{eq:AS'b=ASb}. We note that this inclusion-exclusion principle also plays an important role in the linked-cluster expansion algorithm \cite{link_cluster_book,link_cluster_rev} for computing extensive properties in statistical mechanics.

We adapt the nested commutator bound in Theorem~\ref{thm:exponential_H} to show that $A(S_{\rm b};t)$ decays exponentially with $|S_{\rm b}|$.  The proof of this result is delayed until Sec.~\ref{sec:lemma62proof}.
\begin{lem}\label{lem:ASb<}
    There exist constants $c_d$ determined by $d$, and $c_{\rm box}$ determined by $d,\mu,v_{\rm LR}$, such that Eq.~\eqref{eq:rbox=t} holds, and, for any connected cluster $S_{\rm b}$, \begin{equation}\label{eq:ASb<}
        \norm{A(S_{\rm b};t)} \le c_d \exp\qty[-\mu (3^{-d}|S_{\rm b}| -1) r/5].
    \end{equation}
\end{lem}

We then approximate $A(t)$ by an operator \begin{equation}
    \widetilde{A}(t) = \sum_{\text{connected cluster }S_{\rm b}: |S_{\rm b}|\le M_*} A(S_{\rm b};t)
\end{equation} 
that truncates the sum over clusters in Eq.~\eqref{eq:At=clusters} at size $M_*$ (an integer).
The truncation error is bounded by 
\begin{align}\label{eq:A-tildeA<}
    \norm{A(t)-\widetilde{A}(t)} &\le \sum_{\text{connected cluster }S_{\rm b}: |S_{\rm b}|> M_*} \norm{A(S_{\rm b};t)} \nonumber\\ 
    &\le c_d \sum_{m>M_*} (3^de)^m \exp\qty[-\mu (3^{-d}m -1) r/5] \nonumber\\ 
    &\le c_d\exp\qty[\frac{\mu r}{5}-M_*\qty(3^{-d}\frac{\mu}{5} r-d\log 3)] \frac{1/2}{1-1/2} \le c_d\exp\qty[-\mu r(3^{-d-2}M_*-1)].
\end{align}
Here in the second line, we have applied Lemma \ref{lem:ASb<} and Proposition \ref{prop:clustercounting}
on the number of connected clusters $\le (3^de)^m$ of size $m$. In the last line, we have used the fact that
\begin{equation}
    3^{-d}\frac{\mu}{5} r-d\log 3 \ge \max\left(\log 2, \; 3^{-d-2}\mu r\right)
\end{equation}
from Eq.~\eqref{eq:rbox=t}. 

The theorem then follows by identifying $\widetilde{A}(S;t)$ with $A(S_{\rm b};t)$, where the $S$ are the sites contained in $S_{\rm b}$. Equation \eqref{eq:A-tA_thm} comes from Eq.~\eqref{eq:A-tildeA<} with $M = M_* r^d$ and an updated constant $c_d$. \end{proof}

\subsection{Simulatability bound}

According to the expansion obtained in Theorem \ref{thm:expansion}, one can classically simulate $A(t)$ by its truncation $\widetilde{A}(t)$.

\begin{cor}\label{cor:polynomial_simulate}
Let $H$ be a quasi-local Hamiltonian on a $d$-dimensional square lattice $V$ of qudits, whose Lieb-Robinson bound is given by Proposition \ref{prop:quasilocalbound} with parameters $c_{\rm LR},\mu,v_{\rm LR}$. Consider a local operator $A$ with $\norm{A}=1$ acting on a connected region of $\mathrm{O}(1)$ sites, and an initial state $\rho$ whose marginal $\rho_S$ on any connected set $S\subset V$ can be obtained classically with complexity $\exp[\mathrm{O}(|S|)]$. The expectation $\mathrm{Tr}[A(t)\rho]$ with any $t>0$ can be computed to error $\epsilon$ with complexity \begin{equation}\label{eq:simulation_complex}
        \exp[\mathrm{O}\qty((1+v_{\rm LR}t)^{d-1}\qty (v_{\rm LR}t+\log \frac{1}{\epsilon}) )].
    \end{equation}
\end{cor}

Here the complexity assumption on the initial state applies to e.g. product states, thermal Gibbs states at high temperature \cite{highT_locality14}, and gapped ground states that are connected to solvable states by a gapped path of Hamiltonians \cite{hastings2005}. The reason is that for those states, expectation values of local observables can be computed locally in an efficient way, so that the reduced density matrix $\rho_S$ can be obtained by state tomography with an extra $\exp[\mathrm{O}(|S|)]$ overhead. Note that our algorithm in the proof of Corollary \ref{cor:polynomial_simulate} actually applies to more general initial states, which reduces the problem of computing local observables after time evolution to the problem of obtaining local marginals of the initial state.

\eqref{eq:simulation_complex} achieves a polynomial dependence of $1/\epsilon$ even for $d>1$, which improves upon a quasi-polynomial bound $\exp[\mathrm{O}\qty((v_{\rm LR}t+\log \frac{1}{\epsilon})^{d} )]$ from standard Lieb-Robinson bounds (see e.g. Proposition 4.4 of Ref.~\cite{AnthonyChen:2023bbe}). On the other hand, although a polynomial-in-$1/\epsilon$ bound (Theorem 6 of Ref.~\cite{cluster_Alhambra23}) has been obtained from cluster expansions in the restricted setting where $\rho$ is a product state, its complexity $\exp[\mathrm{O}\qty(\mathrm{e}^{\mathrm{\Theta}(t)}\qty (t+\log \frac{1}{\epsilon}) )]$ grows exponentially faster than Eq.~\eqref{eq:simulation_complex} in terms of the $t$ dependence. The physical reason is that the cluster expansion techniques in Ref.~\cite{cluster_Alhambra23} work for any bounded-degree graph.  Our bound takes full advantage of the additional locality properties of finite-dimensional quantum systems.  Our simulation complexity \eqref{eq:simulation_complex} thus rules out a proposed super-polynomial quantum advantage in analog quantum simulators \cite{stability_error_simu24} while maintaining the correct scaling with $t$.

We argue that bound \eqref{eq:simulation_complex} is tight with respect to the dependence on both $t$ and $1/\epsilon$, because of the following. First, we argue that the $d=1$ case of \eqref{eq:simulation_complex} is tight because the Lieb-Robinson bound \eqref{lrbound1} can be (almost) saturated by a nearest-neighbor Hamiltonian \cite{AnthonyChen:2023bbe}, so the operator is effectively evolved in a truncated subsystem of $\mathrm{\Theta}(v_{\rm LR}t+\log \frac{1}{\epsilon})$ qudits, whose simulation complexity is then exponential in the subsystem size and given in Eq.~\eqref{eq:simulation_complex}. Building on this 1d example, we expect that the following time-dependent protocol $H(s)$ would saturate Eq.~\eqref{eq:simulation_complex} for any constant dimension $d>1$ (note that Theorem \ref{cor:polynomial_simulate} generalizes to time-dependent cases). Here time $s$ is backwards so that operator $A$ first hits $H(0)$ instead of $H(t)$ in Heisenberg evolution. In the first half of the protocol $0<s<t/2$, $H(s)$ just evolves the initial local operator along a one-dimensional chain [e.g. $H(s)$ is a set of decoupled 1d Hamiltonians], so that there is a tail $A_{\rm tail}$ of operator $A(t/2)$ that is supported in a chain of length $\sim v_{\rm LR}t+ \log \frac{1}{\epsilon}$ with \begin{equation}\label{eq:Atail}
    \norm{A_{\rm tail}}\sim \epsilon.
\end{equation}
Such a tail exists due to the tightness of Lieb-Robinson bound in 1d discussed above. The second half of the protocol $t/2<s<t$ then evolves in the horizontal directions of the original chain, which grows $A_{\rm tail}$ to a highly entangled operator $A_{\rm tail}'$ occupying a ``cylinder'' of height $\sim v_{\rm LR}t+\log \frac{1}{\epsilon}$ and radius $\sim v_{\rm LR}t$. This second half can even be a quantum circuit with strict light cones. Since unitary evolution does not change operator norm \eqref{eq:Atail} of this part of operator $\norm{A_{\rm tail}'}\sim \epsilon$, one has to simulate dynamics in the whole cylinder to get precision $\epsilon$, which requires complexity \eqref{eq:simulation_complex} that is an exponential of the cylinder volume.

\begin{proof}[Proof of Corollary \ref{cor:polynomial_simulate}]

In the proof of Theorem \ref{thm:expansion}, we expanded $A(t)$ into connected clusters $S_{\rm b}$. Choosing \begin{equation}\label{eq:mstar=}
    M_*=\left \lfloor \frac{3^{d+2}}{\mu r} \log(\frac{2c_d}{\epsilon}) \right\rfloor+3^{d+2}+1,
\end{equation}
the truncation error \eqref{eq:A-tildeA<} leads to
\begin{align}
    \qty|\mathrm{Tr}\qty[\rho\qty(A(t)-\widetilde{A}(t))]| &\le \norm{A(t)-\widetilde{A}(t)} \le \frac{\epsilon}{2}.
\end{align}

Therefore, to simulate $\mathrm{Tr}[\rho A(t)]$, it suffices to compute \begin{align}
    \mathrm{Tr}[\rho \widetilde{A}(t)] &= \sum_{\text{connected cluster }S_{\rm b}: |S_{\rm b}|\le M_*} a(S_{\rm b}),
\end{align}
where \begin{equation}
    a(S_{\rm b}):=\mathrm{Tr}[\rho A(S_{\rm b};t)]=\mathrm{Tr}[\rho_{S_{\rm b}} A(S_{\rm b};t)]
\end{equation}
only involves a marginal of the initial state $\rho_{S_{\rm b}}:=\rho_{\boxx_0\cup \boxx_1\cup\cdots\cup \boxx_m}$ for $S_{\rm b}=\{\boxx_0,\boxx_1,\cdots,\boxx_m\}$. 
According to Eq.~\eqref{eq:AS'b=ASb},
\begin{equation}\label{eq:a=tildea}
    a(S_{\rm b}) = \widetilde{a}(S_{\rm b}) - \sum_{S'_{\rm b}\subsetneqq S_{\rm b}:\, S'_{\rm b}\neq \emptyset} a(S'_{\rm b}),
\end{equation}
where \begin{equation}\label{eq:tildea=}
    \widetilde{a}(S_{\rm b}):=\mathrm{Tr}[\rho_{S_{\rm b}} e^{t\mathcal{L}_{S_{\rm b}}}A]=\mathrm{Tr}[\rho_{S_{\rm b}} e^{\mathrm{i} tH_{S_{\rm b}}}Ae^{-\mathrm{i} tH_{S_{\rm b}}}].
\end{equation}

The simulation algorithm then works as follows:

\begin{algorithm}[H]
\caption{A classical simulation algorithm for $\mathrm{Tr}[\rho A(t)]$}
\begin{algorithmic}
\State Generate all connected clusters $S_{\rm b}\ni \boxx_0$ with $|S_{\rm b}|\le M_*$
\State Create array $a(S_{\rm b}) \gets 0$ for the generated clusters $\{S_{\rm b}\}$
\For{$m=1,2,\cdots,M_*$}
    \ForAll{connected cluster $S_{\rm b}\ni \boxx_0$ with $|S_{\rm b}|=m$}
    \State Compute $\widetilde{a}(S_{\rm b})$ in Eq.~\eqref{eq:tildea=} by exponentiating matrix $H_{S_{\rm b}}$
    \State Update $a(S_{\rm b})$ by Eq.~\eqref{eq:a=tildea} using $\widetilde{a}(S_{\rm b})$ and the stored values for $a(S'_{\rm b})$ where $|S'_{\rm b}|<m$
    \EndFor
\EndFor
\State \Return $\sum_{S_{\rm b}} a(S_{\rm b})$
\end{algorithmic}
\end{algorithm}

For each $\widetilde{a}(S_{\rm b})$, since it only involves a subsystem of $\le M_*r^d$ qudits where the marginal $\rho_S$ is assumed to be computable in exponential time, $\widetilde{a}(S_{\rm b})$ can be computed accurately by \begin{equation}\label{eq:expmr}
   N = \exp[\mathrm{O}(M_*r^d)]
\end{equation}  
resources. This dominates the complexity for $a(S_{\rm b})$, as the sum over $a(S'_{\rm b})$ in Eq.~\eqref{eq:a=tildea} contains at most $\exp[\mathrm{O}(M_*)]$
terms\footnote{Note that one can be more clever to treat this sum by storing intermediate results. However, this does not change the scaling of complexity for the whole algorithm.}. Since there are at most $\exp[\mathrm{O}(M_*)]$ clusters and they can be enumerated by a similar cost according to Proposition \ref{prop:clustercounting}, the total complexity of the algorithm is still of the scaling \eqref{eq:expmr}, which becomes Eq.~\eqref{eq:simulation_complex} by plugging in Eqs.~\eqref{eq:rbox=t} and \eqref{eq:mstar=}. The final error is bounded by $\epsilon$ because the $M_*$ truncation error is bounded by $\epsilon/2$ in Eq.~\eqref{eq:A-tildeA<}, and the other operations like computing $\widetilde{a}(S_{\rm b})$ are essentially exact using resources above.
\end{proof}

\subsection{Applying the nested commutator bound: proof of Lemma \ref{lem:ASb<}}
\label{sec:lemma62proof}

\begin{proof}
For $S_{\rm b}$ that only contains boxes $b_0$ or its nearest neighbors (there are $3^d-1$ of them), 
Eq.~\eqref{eq:ASb<} holds by choosing a sufficiently large $c_d$ determined by $d$. The reason is that $\norm{A(\{b_0\};t)}\le 1$ due to Eq.~\eqref{eq:Ab0=}, and the other $\norm{A(S_{\rm b};t)}$ can be bounded by iterating Eq.~\eqref{eq:AS'b=ASb}: for example, let $b_1$ be one neighbor of $b_0$; then \begin{equation}
    \norm{A(\{b_0,b_1\};t)} = \norm{e^{t\mathcal{L}_{\{b_0,b_1\}}} A - A(\{b_0\};t)} \le \norm{e^{t\mathcal{L}_{\{b_0,b_1\}}} A} + \norm{ A(\{b_0\};t)} \le 1+1=2.
\end{equation} 
There are only finitely many $S_{\rm b}$ that only contain $b_0$ and its neighbors, so, by a finite number of iterations, they all satisfy $\norm{A(S_{\rm b};t)}\le c_d$ for some constant $c_d$ determined by $d$, so that \eqref{eq:ASb<} is satisfied for them because the exponential factor is always $\le 1$ for these $S_{\rm b}$ with $|S_{\rm b}|\le 3^d$.

Beyond these finitely many connected clusters, any other $S_{\rm b}$ contains at least one box that is a distance $\ge r$ from the initial operator $A$. We can then find boxes $S_1,S_2,\cdots, S_M \in S_{\rm b}$ with $M\ge 1$ such that each $S_m$ is surrounded by a larger box $B_m$ of side length $2(r-\chi)$ (where $\frac{r}{2}-\chi$ is an integer and $r,\chi\ge 1$ are chosen shortly) and the larger boxes $B_1,\cdots,B_M$ are a distance $\ge 2\chi$ from each other and from $b_0$. See Fig.~\ref{fig:box} for an illustration.  More precisely, because we have excluded the finitely many $S_{\rm b}$ in the neighborhood of $b_0$, we can always first find an $S_1\in S_{\rm b}$ whose larger box $B_1$ does not touch $b_0$. Then we try to find an $S_2\in S_{\rm b}$ that is not a neighboring box of $S_1$, and so on. We can always find such a set $S_1,S_2,\cdots, S_M \in S_{\rm b}$ containing at least
\begin{equation}\label{eq:M>Sb}
    M\ge 3^{-d}|S_{\rm b}|-1
\end{equation}
boxes, because each $S_m$ only forbids its $3^d-1$ neighbors to be selected at the same time. Here we need the $-1$ in Eq.~\eqref{eq:M>Sb} to also avoid boxes neighboring $b_0$.

\begin{figure}[t]
    \centering
    \includegraphics[width=0.5\linewidth]{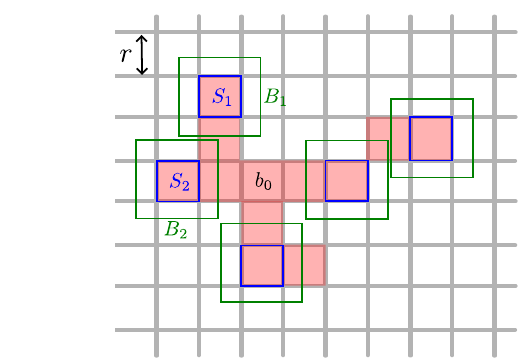}
    \caption{To apply the nested commutator bound for a connected cluster $S_{\rm b}$ (red shaded region), we find boxes $S_1,S_2,\cdots,S_M\in S_{\rm b}$ (blue) that are well separated from each other and from the initial operator in $b_0$. We can always find $M=\mathrm{\Theta}(|S_{\rm b}|)$ number of them, which are surrounded by larger boxes $B_1,\cdots$ that we use in the proof of the nested commutator bound. }
    \label{fig:box}
\end{figure}

For a given $S_{\rm b}$, after finding these $M$ boxes $S_m$ together with their surrounding boxes $B_m$, we can then apply Theorem \ref{thm:exponential_H} to bound $\norm{A(S_{\rm b};t)}$. The reason is that although Theorem \ref{thm:exponential_H} is stated directly for the nested commutator $C^R_{S_1,\cdots,S_M}(t)$, its proof works by individually bounding each of the irreducible skeleton that contribute to $C^R_{S_1,\cdots,S_M}(t)$. Here $\norm{A(S_{\rm b};t)}$ is also bounded by a sum over irreducible skeletons, with the only difference that here the causal forests in Eq.~\eqref{eq:At=clusters} corresponding to the irreducible skeletons are restricted to be contained in $S_{\rm b}$. In other words, $\norm{A(S_{\rm b};t)}$ is bounded by the bound on $C^R_{S_1,\cdots,S_M}(t)$ if $H$ only contains terms supported inside $S_{\rm b}$; but this truncated $H$ is still a quasi-local Hamiltonian so we can just apply Eq.~\eqref{eq:exponential_H_C<}. Another small difference here is that $\{B_m\}$ do not lie in a square lattice. Theorem \ref{thm:exponential_H} generalizes to this case because its proof only uses the fact that the whole space can be tiled by $\{B_m\}$ and some other regions where the connectivity graph of the regions have bounded degree $\le 3^d$; this property also holds here. 

Equation \eqref{eq:exponential_H_C<} thus implies that
\begin{align}
    \norm{A(S_{\rm b};t)} &\le \qty(c'(t+1)\cdot r^{2(d-1)}\cdot e^{\mu(v_{\rm LR}t-r/2+\chi)})^M \nonumber\\
    &\le \qty(c''  r^{2d-1} \cdot e^{-\mu r/4})^M \le \qty(e^{-\mu r/5})^M, \label{eq:A<comm_bound}
\end{align}
where $c',c''$ are constants. Here in the first line, we have used $|\partial S_i|\le |\partial B_i|=\mathrm{O}(r^2)$ and chosen a sufficiently large constant $\chi$ independent of $t$ to get the factor $(t+1)$. In the second line of Eq.~\eqref{eq:A<comm_bound}, we have used Eq.~\eqref{eq:rbox=t} with a sufficiently large constant $c_{\rm box}$, so that $t+1=\mathrm{O}(r)$ and $v_{\rm LR}t-r/2+\chi=-r/4 - c_{\rm box}+\chi\le -r/5 - \mu^{-1}\log(c'' r^{2d-1})$; this last condition is achievable because the function $f(r):=r/20 - \mu^{-1}\log(c'' r^{2d-1})$ is bounded from below. Equation \eqref{eq:ASb<} follows by plugging Eq.~\eqref{eq:M>Sb} in Eq.~\eqref{eq:A<comm_bound}.
\end{proof}

\section{Spontaneous symmetry-breaking phases of finite symmetries \label{sec:ssb}}
We now turn to the second main application of our nested commutator bounds, and discuss the diagnosis of non-trivial phases of quantum matter.  We focus on phases that exhibit spontaneous symmetry breaking (SSB) of finite symmetries. SSB can be detected by the presence of ground states with exponentially small energy splitting, along with order and disorder parameters computed in these nearly degenerate states. For example, in the Ising ferromagnetic phase, there are two degenerate (in the thermodynamic limit) ground states $\ket{\psi_{\pm}}$ with eigenvalues $\pm 1$ under the global $\mathbb{Z}_2$ symmetry operator. These states are adiabatically connected to the $\pm$ GHZ states 
\begin{equation}
    \ket{\pm}=\frac{\ket{\mathbf{0}}\pm\ket{\mathbf{1}}}{\sqrt{2}}.  \label{eq:ghz}
\end{equation} 

More concretely, it was shown in Ref.~\cite{hastings2005} using Lieb-Robinson bounds that the splitting between the energies of $\ket{\psi_{\pm}}$ is at most exponentially small in the linear system size $L$: \begin{equation}
    \delta := \left|\langle \psi_+ |H|\psi_+\rangle - \langle \psi_-|H|\psi_-\rangle \right| \le \mathrm{e}^{-\mathrm{O}(L)}. \label{eq:energysplitting}
\end{equation} However, deep in the ferromagnetic phase, one expects that adding a small transverse field $h\sum_iX_i$ would lead to $\delta \sim \mathrm{e}^{-\mathrm{O}(L^d)}$, because the two ground states $\ket{\pm}$ are distinguished only by the global operator $\prod_iX_i$. Intuitively, this expectation is because it requires going to $\mathrm{O}(L^d)$, not $\mathrm{O}(L)$, in perturbation theory, in order to connect $|\pm\rangle$ via perturbations. 

Spontaneous symmetry breaking is also marked by a long-range order parameter $\lim_{|i-j|\to\infty}\bra{\psi}Z_iZ_j\ket{\psi}\sim\mathcal{O}(1)$ in states in the ground-state subspace, as well as a quickly decaying disorder parameter $\lim_{R\to\infty}\bra{\psi}D_{R}\ket{\psi}$. Here, $d$ is the spatial dimension, and $D_{R}=\prod_{i\in B_{R}(v)} X_i$ is the global $\mathbb{Z}_2$ symmetry operator restricted to a $d$-dimensional ball of radius $R$ centered at any vertex $v$. Like with the splitting, conventional Lieb-Robinson bounds give a far looser bound, saying that $\lim_{R\to\infty}\langle \psi_+|D_R|\psi_+\rangle \lesssim \mathrm{e}^{-\mathrm{O}(R)}$ rather than the bound one might expect from perturbation theory: $\lim_{R\to\infty}\langle \psi_+|D_R|\psi_+\rangle \lesssim \mathrm{e}^{-\mathrm{O}(R^d)}$.  Indeed, this latter scaling is observed numerically \cite{zhao2021}.  
\subsection{Proofs of volume-law scaling}
Using our nested commutator bound, we will confirm the intuition raised in the two arguments above.  If states $\ket{\psi_{\pm}}$ are connected to the GHZ fixed point states $\ket{\pm}$ by finite-time evolution generated by a quasi-local Hamiltonian (satisfying Definition~\ref{defquasilocal}), then the ground state splitting is in fact at most exponentially small in the \emph{volume} of the system. Similarly, we show that the disorder parameter is also at most exponentially small in the \emph{volume} of $B_{\ell}(v)$. 
It was shown in Ref.~\cite{yin2024}  that, for a sufficiently small local perturbation in the vicinity of the ferromagnetic fixed point, such a quasi-local generator exists. The usual quasi-adiabatic generator \cite{hastings2004,osborne2007simulating,bachmann2012automorphic} that can be constructed within the entire adiabatic phase decays almost exponentially in diameter [to be precise, exponentially in $R/\log^2(R)$], so outside of this vicinity of the fixed point, it is not known whether an adiabatic generator satisfying our stronger definition of quasi-local exists.
Therefore, at least for states in the vicinity of the ferromagnetic fixed point, our nested commutator results give bounds on the ground-state energy splitting and disorder parameter that are much tighter than those given by the usual Lieb-Robinson bounds. It seems reasonable to conjecture that such a quasi-local generator might exist within the entire phase, not just in the vicinity of the fixed point. This would imply that these bounds hold within the entire ferromagnetic phase.

In the following, we specialize to properties of the Ising ferromagnetic phase on a lattice in $d$ dimensions, with a constant finite density of sites with respect to the Euclidean distance, where each site hosts a qubit. However, our results are straightforward to generalize to spontaneous symmetry-breaking phases of other finite symmetries, in systems with more general local Hilbert spaces.  The following two theorems prove the two conjectures stated above, given the assumption of the exponential-in-volume-tailed quasi-adiabatic generator.

\begin{thm}\label{thm:splitting}
If $\ket{\psi_{\pm}}$ are ground states of a gapped, quasi-local Hamiltonian satisfying Def.~\ref{defquasilocal}, that are connected to the $\pm$ GHZ states $\ket{\pm}$ defined in Eq.~\eqref{eq:ghz} in spatial dimension $d$ by finite-time evolution generated by a quasi-local Hamiltonian, then the splitting $\delta$ defined in Eq.~\eqref{eq:energysplitting} obeys 
    \begin{align}
        \delta\leq cL^d \mathrm{exp}\left(-\gamma\frac{(L-vt)^d}{(vt)^{d-1}}\right),
    \end{align}
    \label{thm:split}
    for $L\gg vt$, where $c,\gamma$ are $\mathcal{O}(1)$ constants, $L$ is the linear system size, and $v$ and $t$ are $\mathcal{O}(1)$ constants describing the finite-time evolution and its quasi-local generator.
\end{thm}

\begin{thm}\label{thm:disorder}
If $\ket{\psi}=U(t)\ket{\mathbf{0}}$ where $U(t)$ describes finite-time evolution in $d$ dimensions generated by a local, $\mathbb{Z}_2$ symmetric Hamiltonian, then the disorder parameter $\lim_{R\to\infty}\langle\psi|D_R|\psi\rangle$ for a $d$-dimensional ball of radius $R$ is upper bounded by
    \begin{align}
        \lim_{R\to\infty}\bra{\psi}D_R\ket{\psi} \leq c_{\mathrm{LR}}\mathrm{exp}\left(-\gamma\frac{(R-vt)^d}{(vt)^{d-1}}\right),
    \end{align}
    \label{thm:mot}
    where again $c_{\mathrm{LR}}$ and $\gamma$ are $\mathcal{O}(1)$ constants and $v,t$ describe the time evolution.
\end{thm}

One subtlety about the above theorems is that in any given state $|t|$ may be very large, even $\mathcal{O}(\log L)$. Therefore, to probe the scaling of the disorder parameter and truly rule out the possibility that a state is connected to the ferromagnetic fixed point by finite [or $\mathcal{O}(\log L)$] time evolution generated by a quasi-local Hamiltonian, we need to compute expectation values of extensive operators and compare very small values (exponentially small in $L$ versus exponentially small $L^d$). In Sec.~\ref{srkstates}, we present an example of a state $|\psi(\beta)\rangle$ that fails to satisfy Theorem~\ref{thm:mot}. However, due to the subtlety above, we can obtain another state $|\tilde{\psi}(\beta)\rangle$ whose overlap with $|\psi(\beta)\rangle$ goes to 1 in the thermodynamic limit, and does not violate Theorem~\ref{thm:mot}. It is likely that $|\tilde{\psi}(\beta)\rangle $ can be connected to the ferromagnetic fixed point if we allow $|t|$ to be $\mathcal{O}(\log L)$. This point is discussed more in depth in Ref.~\cite{sahay20252}. However, for realistic systems where $t$ is an $\mathcal{O}(1)$ constant, the scaling of the disorder parameter provides a reliable diagnostic for the ferromagnetic phase.

To prove the above two theorems, we make use of the following lemma. The rough idea of the lemma is that an expectation value of the form $\bra{\psi}D_{\mathrm{tot}}O\ket{\psi}$, where $D_{\mathrm{tot}}=\prod_iX_i$ is the global $\mathbb{Z}_2$ symmetry operator, can be written as the expectation value of a large nested commutator of various single-site unitaries $Z_i$ with $O(t)$. Here, $\ket{\psi}=U(t)\ket{\mathbf{0}}$, where $U(t)$ is $\mathbb{Z}_2$ symmetric, and $O(t)=U(t)^\dagger O(0)U(t)$. The polarized state is denoted by $\ket{\mathbf{0}}$, which is the unique eigenstate of every $Z_i$ operator with eigenvalue 1. This rewriting of $\langle\psi|D_{\mathrm{tot}}O|\psi\rangle$ will allow us to apply our nested commutator bound.

\begin{lem}
    Let $O$ be an operator and let $|\psi\rangle=U(t)|0\rangle$. Then for any subset of sites $v_1, \dots, v_m$, we have
    \begin{equation}
        \bra{\psi} D_{\mathrm{tot}}O \ket{\psi} = \frac{1}{2^m}\bra{\mathbf{0}}D_{\mathrm{tot}}[[\dots[O(t),Z_{v_1}],\dots Z_{v_{m-1}}], Z_{v_{m}}]\ket{\mathbf{0}} \equiv \frac{1}{2^m}\bra{\mathbf{0}}D_{\mathrm{tot}}[O(t), Z_{\vec v}]\ket{\mathbf{0}}.
    \end{equation}
    \label{lem:nest}
\end{lem}
\begin{proof}
We use the following observations: (\emph{1}) the evolution operator $U(t)$ commutes with $D_{\mathrm{tot}}$, (\emph{2}) $\ket{\mathbf{0}}$ is an eigenstate of $Z_{v_i}$, and (\emph{3}) $Z_{v_i}$ anticommutes with $D_{\mathrm{tot}}$. Putting these observations together, we get
    \begin{align}\label{nest1}
        \bra{\psi}D_{\mathrm{tot}}O\ket{\psi} &= \bra{\mathbf{0}}D_{\mathrm{tot}}O(t)\ket{\mathbf{0}} \notag\\
               &= \bra{\mathbf{0}}Z_{v_1}D_{\mathrm{tot}}O(t)\ket{\mathbf{0}} \notag\\
               &= -\bra{\mathbf{0}}D_{\mathrm{tot}}Z_{v_1}O(t)\ket{\mathbf{0}} \notag\\
             &= \bra{\mathbf{0}}D_{\mathrm{tot}}[O(t),Z_{v_1}]\ket{\mathbf{0}}-\bra{\psi}D_{\mathrm{tot}}O\ket{\psi}.
    \end{align}
    Rearranging the above gives $\bra{\psi} D_{\mathrm{tot}}O \ket{\psi} = \frac{1}{2}\bra{\mathbf{0}}D_{\mathrm{tot}}[O(t),Z_{v_i}]\ket{\mathbf{0}}$. This is the base case. Now we assume that Eq.~\ref{lem:nest} holds for a subset of sites $\vec v'=(v_1,\cdots, v_{m-1})$, and use the exact same reasoning to make the inductive step:
    \begin{align}
        \bra{\psi}D_{\mathrm{tot}}O\ket{\psi}&=\frac{1}{2^{m-1}}\bra{\mathbf{0}}D_{\mathrm{tot}}[O, Z_{\vec v'}]\ket{\mathbf{0}} \notag
        \\ &= 
        -\frac{1}{2^{m-1}}\bra{\mathbf{0}}D_{\mathrm{tot}}Z_{v_{m}}[O, Z_{\vec v'}]\ket{\mathbf{0}} \notag\\
        &= \frac{1}{2^{m-1}}\left(\bra{\mathbf{0}}D_{\mathrm{tot}}[[O, Z_{\vec v'}], Z_{v_{m}}]\ket{\mathbf{0}} - \bra{\mathbf{0}}D_{\mathrm{tot}}[O, Z_{\vec v'}]\ket{\mathbf{0}}\right)\notag\\
        &=\frac{1}{2^m}\bra{\mathbf{0}}D_{\mathrm{tot}}[O, Z_{\vec v}]\ket{\mathbf{0}}.
    \end{align}
Here $\vec v=(v_1,\cdots,v_{m-1},v_m)$ and we define $[O, Z_{\vec v}] \equiv [[\dots[O,Z_{v_1}],\dots], Z_{v_m}]$. This completes the proof.
\end{proof}

With this technical lemma in hand, we now prove the theorems.

\begin{proof}[Proof of Theorem \ref{thm:split}]
  Using $\ket{\psi_{\pm}}=U(t)\ket{\pm}$, we obtain
\begin{align}
    \delta &=|\bra{+}H(t)\ket{+}-\bra{-}H(t)\ket{-}| \notag \\
    &=2\left|\mathrm{Re}\left\{\bra{\mathbf{1}}H(t)\ket{\mathbf{0}}\right\}\right| \notag \\
    &=2\left|\mathrm{Re}\left\{\bra{\mathbf{0}}D_{\mathrm{tot}}H(t)\ket{\mathbf{0}}\right\}\right| \notag \\
    &\leq 2\sum_{u\in V}\sum_{X \ni u}\left|\mathrm{Re}\left\{\bra{\mathbf{0}}D_{\mathrm{tot}}H_X(t)\ket{\mathbf{0}}\right\}\right|,
\end{align}
where in the last line, we decomposed the $H(t)$ into terms on connected clusters $X$ as in (\ref{defquasilocal}). Note that we did not need to assume that $U(t)$ is $\mathbb{Z}_2$ symmetric here because we did not need to apply the first step of (\ref{nest1}). 

If $H$ is strictly local, then we just have a sum over $u\in V$ and each of the terms above is of the form $\bra{\mathbf{0}}D_{\mathrm{tot}}O(t)\ket{\mathbf{0}}$, so we can apply Lemma~\ref{lem:nest} and then Theorem~\ref{thm:exponential_H} to bound each term. Then Theorem~\ref{thm:splitting} follows directly from Corollary~\ref{cor:volumescaling_exp}. There is an overall factor of $L^d$ coming from the sum over $u\in V$.

If $H$ is quasi-local, then we divide the sum into clusters $X$ of size $|X|<\frac{L^d}{(B\xi)^{d-1}}$ and $|X|\geq \frac{L^d}{(B\xi)^{d-1}}$, where $B\xi$ is the linear size of the $m$ boxes that we divide the lattice into in Corollary~\ref{cor:volumescaling_exp}, to fit in $m$ balls containing the operators $Z_{v_1},\dots Z_{v_m}$. Roughly speaking, for every $u\in V$, we sum over the connected clusters $X$ containing $u$ [of which there are $\leq (\mathrm{e}\Delta)^{|X|}$ by Proposition~\ref{prop:clustercounting}] suppressed by their weight $h\mathrm{e}^{-\kappa |X|}$ and further suppressed by a nested commutator obtained by inserting as many balls as possible in a region of size $\sim L^d-|X|$. More precisely, the nested commutator can include all of the operators $Z_{v_1},\ldots ,Z_{v_m}$ except for those in boxes that overlap with $X$. To accommodate for the worst-case scenario where the connected cluster $X$ forms a net along the edges of the boxes of linear size $B\xi$ (see the definitions in Corollary~\ref{cor:volumescaling_exp}), to touch as many boxes as possible for smallest $|X|$, we use a region of size $ L^d-(B\xi)^{d-1}|X|$ rather than $L^d-|X|$. Once $|X|\sim \frac{L^d}{(B\xi)^{d-1}}$, the worst-case scenario gives no further suppression from the nested commutators. Putting the above observations together, we get
\begin{align}
    \begin{split}
        \delta&\leq 2hL^d\left(\sum_{\substack{X\ni u\\ |X|<L^d/(B\xi)^{d-1}}}(\mathrm{e}\Delta)^{|X|}\mathrm{e}^{-\kappa |X|}\mathrm{e}^{-\frac{C_2\mu}{2\alpha^{d-1}}\frac{L^d-(B\xi)^{d-1} |X|}{(vt)^{d-1}}}+\sum_{\substack{X\ni u\\ |X|\geq L^d/(B\xi)^{d-1}}}(\mathrm{e}\Delta)^{|X|}\mathrm{e}^{-\kappa |X|}\right)\\
        &\leq 2hL^d\left(\mathrm{e}^{-\frac{C_2\mu}{2\alpha^{d-1}}\frac{L^d}{(vt)^{d-1}}}\sum_{\substack{X\ni u\\ |X|<L^d/(B\xi)^{d-1}}}\left(\Delta \mathrm{e}^{1-\kappa }\mathrm{e}^{\frac{C_2\mu}{2\alpha^{d-1}}\frac{(B\xi)^{d-1}}{(vt)^{d-1}}}\right)^{|X|}+\frac{(\Delta \mathrm{e}^{1-\kappa})^{\frac{L^d}{(B\xi)^{d-1}}}}{1-\Delta \mathrm{e}^{1-\kappa}}\right).\\
    \end{split}
\end{align}

Performing the first sum and simplifying notation using $y=\Delta \mathrm{e}^{1-\kappa}$, we have 
\begin{equation}
        \delta\leq 2hL^d\left(\frac{y^{\frac{L^d}{(B\alpha v t)^{d-1}}}-\mathrm{e}^{-\frac{C_2\mu}{2\alpha^{d-1}}\frac{L^d}{(vt)^{d-1}}}}{y\mathrm{e}^{\frac{C_2\mu}{2\alpha^{d-1}}\frac{(B\xi)^{d-1}}{(vt)^{d-1}}}-1}+\frac{y^{\frac{L^d}{(B\alpha vt)^{d-1}}}}{1-y}\right),
    \end{equation}
where we used $\xi=\alpha v t$ as in the proof of Corollary~\ref{cor:volumescaling_exp}. Since $y<1$ by \eqref{eq:kappalowerbound}, every term is exponentially decaying in $L^d$, and we can use the same steps as in Corollary~\ref{cor:volumescaling} to obtain Theorem~\ref{thm:splitting} with a modified $\gamma$ and $c$ from the strictly local case.
\end{proof}

\begin{proof}[Proof of Theorem \ref{thm:mot}]  
We make the observation that the disorder operator can be expressed as $D_R(t) = D_{\mathrm{tot}}\bar{D}_{R}^\dagger(t)$, where $\bar{D}_{R}^\dagger$ is the Hermitian conjugate of the disorder operator for the complement of $B_{R}(v)$. This follows from the fact that $U(t)$ commutes with $D_{\mathrm{tot}}=D_R\bar{D}_{R}$, the global $\mathbb{Z}_2$ symmetry operator. 
Now we can apply Lemma~\ref{lem:nest} to fill the support of $D_R$, $B_R$, with operators at site $v_i$ in balls $B_i$ (see Fig.~\ref{fig:setup}). Note that $\bar{D}_{R}^\dagger$ is supported only in the complement of $B_{R}$. Substituting $D_{\mathrm{tot}}\bar{D}_{R}^\dagger(t)$ for $D_R(t)$, we can then apply Lemma~\ref{lem:nest} to get
\begin{align}\label{nestdis}
    \vert \bra{\psi}D_{R}\ket{\psi}\vert \leq \frac{1}{2^m}\left\Vert \left[Z_{\vec v}, \bar{D}_{R}^\dagger (t)\right]\right\Vert \leq C_{\vec v}^{B_R^c}(t),
\end{align}
where we used a slight simplification due to the observation that $\Vert D_{\mathrm{tot}} \Vert = 1$. We can then apply Corollary~\ref{cor:volumescaling_exp} to obtain Theorem~\ref{thm:disorder} for sufficiently large $R$. 
\end{proof}
Note that the usual Lieb-Robinson bounds, applied to Eq.~\eqref{nestdis} with a single commutator upper bounds $\bra{\psi}D_R\ket{\psi}$ by a quantity exponentially small in the diameter of $S$. In dimension $d > 1$, we get a much stronger bound by using a nested commutator, which is exponentially small in the volume of $B_{R}(v)$. From the proof above, we see an illustration of the flexibility to optimize the parameters used to get the bound. In $d = 1$, the optimal choice of commutator is simply a single commutator with an operator $Z_v$ at the center of the interval of radius $R$.

\subsection{Rokhsar-Kivelson states}\label{srkstates}
Rokhsar-Kivelson (RK) states are quantum ground states that encode classical partition functions \cite{rokhsar1988,ardonne2004,castelnovo2008}. These states have been studied in the context of conformal quantum critical points (when the classical partition function goes through a critical point) and have recently received renewed interest in the context of decohered topological order \cite{bao2023}. In the latter context, RK states appear naturally in the thermofield double (Choi state) representation of the decohered density matrix.

Consider the Ising RK wavefunction:
\begin{equation}
    |\psi(\beta)\rangle=\frac{1}{\mathcal{N}}\prod_{\langle ij\rangle}\mathrm{e}^{\beta Z_iZ_j/2}\ket{+},
\end{equation}
where $\mathcal{N}$ is a normalization factor to ensure that $\langle\psi(\beta)|\psi(\beta)\rangle=1$. Here, the product is over all nearest-neighbor vertices on the $d$-dimensional hypercubic lattice. Although $|\psi(\beta)\rangle$ is a $d$-dimensional quantum state, it also encodes a $d$-dimensional classical partition function. It is easy to see that $Z_iZ_j$ correlation functions in $|\psi\rangle$ can be identified with spin-spin correlation functions in the classical Ising model. It follows that
\begin{equation}
\langle \psi(\beta)|Z_iZ_j|\psi(\beta)\rangle\sim\begin{cases}
            \mathrm{exp}(-\mu|i-j|), \beta<\beta_c\\
            \mathcal{O}(1), \beta>\beta_c
           \end{cases},
\end{equation}
where $\beta_c$ is the classical Ising critical temperature and $\mu$ is a finite inverse correlation length. For example, for $d=2$, $\beta_c=\log (1+\sqrt{2})/2$. Because there is a long-range order parameter for $\beta>\beta_c$, one might suspect that $|\psi(\beta>\beta_c)\rangle$ belongs in the $\mathbb{Z}_2$ ferromagnetic phase. However, 
\begin{equation}
    \langle\psi(\beta)|D_{R}|\psi(\beta)\rangle\sim \mathrm{e}^{-2\beta cR^{d-1}},
\end{equation}
where $cR^{d-1}$ is the surface area of the radius $R$ ball in $d$ dimensions. The above scaling of the disorder parameter holds for all values of $\beta$, including for $\beta>\beta_c$, which can be calculated by relating the disorder parameter to a ratio of 2d classical Ising partition functions, where the only difference between the partition functions comes from deleting the spin-spin couplings along $\partial R$ in the numerator. The scaling is then expected from extensivity of free energy $Z=e^{\beta F}$. Further details can be obtained from Appendix A2 of Sahay et al.\cite{sahay20251}. Therefore, this state violates the bound in Theorem~\ref{thm:disorder} for all finite $\beta$, in any dimension. We therefore conclude that $|\psi(\beta)\rangle$, despite the fact that it demonstrates a long-ranged order parameter for all $\beta>\beta_c$, is not in the vicinity of the ferromagnetic fixed point (where a quasi-adiabatic generator satisfying Def.~\ref{defquasilocal} is guaranteed by \cite{yin2024}) for any finite $\beta>\beta_c$.  We conjecture that this state is not in the ferromagnetic phase in any $d$; an independent argument for this result in $d=2$ is found in \cite{sahay20251}.

\section{Conclusion}
This paper has explored the locality of the time evolution of local operators, focusing on obtaining strong tail bounds on large operators supported on a volume $V \gg (v_{\mathrm{LR}}t)^d$---the volume of the Lieb-Robinson light cone.  We found that, loosely speaking, such operators are suppressed as $\exp[-V/(v_{\mathrm{LR}}t)^{d-1}]$, closing a conceptual gap between the Lieb-Robinson bound \cite{Lieb1972,AnthonyChen:2023bbe} and bounds from cluster expansions \cite{cluster_Alhambra23}.   Two immediate applications of such bounds were presented---one on the efficiency of classical simulations of quantum dynamics, and one on the classification of quantum phases of matter.  

We hope that many further applications of these strong volume-tailed Lieb-Robinson bounds can be identified.   A natural future direction is to generalize our volume-tailed bounds to systems with bosons \cite{boson_anharmonic08,boson_empty11,boson_spin13,boson_kuwahara21,boson_finitespeed22,boson_lemm22,boson_empty22,boson_algebraic24} or power-law interactions \cite{power_dyn_area17,power_chen19,power_simu19,power_KSLRB20,lucasprx2020,power_GHZ21,power_yifan21,power_KSOTOC21,power_LRB21,Wprotocol_gorshkov20,power_all2all24}.  It is of great interest to generalize the volume-tailed bounds on the quasi-adiabatic generator, obtained in Ref.~\cite{yin2024} in the vicinity of a ``code" fixed point, to the entire phase of matter.   Lastly, it would also be intriguing if these heavy-tailed Lieb-Robinson bounds could help to prove the entanglement area-law for gapped phases in $d>1$ dimensions, following recent progress \cite{Anshu:2020xkf}.

\section*{Acknowledgements}
C.Y. thanks Alvaro Alhambra for pointing out the connection with Ref.~\cite{stability_error_simu24}. C.Z. thanks Rahul Sahay, Ruben Verressen, and Curt von Keyserlingk for collaboration on related work \cite{sahay20251,sahay20252}, and Michael Levin for helpful discussions. C.Y. and A.L. were supported by the Department of Energy under Quantum Pathfinder Grant DE-SC0024324. C.Z. was supported by the Harvard Society of Fellows and the Simons Collaboration on Ultra Quantum Matter.

\begin{appendix}
\end{appendix}

\bibliography{thebib}

\begin{thebibliography}{65}
\bibitem [{\citenamefont {Sahay}\ \emph {et~al.}(to appear)\citenamefont
  {Sahay}, \citenamefont {Zhang}, \citenamefont {von Keyserlingk},\ and\
  \citenamefont {Verresen}}]{sahay20252}%
  \BibitemOpen
  \bibfield  {author} {\bibinfo {author} {\bibfnamefont {Rahul}\ \bibnamefont
  {Sahay}}, \bibinfo {author} {\bibfnamefont {Carolyn}\ \bibnamefont {Zhang}},
  \bibinfo {author} {\bibfnamefont {Curt}\ \bibnamefont {von Keyserlingk}}, \
  and\ \bibinfo {author} {\bibfnamefont {Ruben}\ \bibnamefont {Verresen}},\
  }\href@noop {} {\enquote {\bibinfo {title} {Is the decohered toric code
  long-range entangled beyond its error threshold?}}\ } (\bibinfo {year} {to
  appear})\BibitemShut {NoStop}%
\end{thebibliography}

\end{document}